\documentclass[lettersize]{IEEEtran}
\usepackage{color}
\usepackage{verbatim}
\usepackage{amsfonts}
\usepackage{amssymb}
\usepackage{stfloats}
\usepackage{cite}
\usepackage{graphicx}
\usepackage{psfrag}
\usepackage{subfigure}
\usepackage{amsmath}
\usepackage{array}
\usepackage{epstopdf}
\usepackage{authblk}
\usepackage{graphicx} 
\usepackage{amsthm} 
\usepackage{lipsum}
\usepackage{verbatim} 
\usepackage{authblk}
\usepackage{mathtools}
\usepackage{cuted}
\usepackage[lined,boxed,ruled]{algorithm2e}
\usepackage{algpseudocode}
\usepackage{framed} 
\usepackage{subfigure}
\usepackage{soul}
\usepackage{bm}
\usepackage{setspace}
\usepackage{hyperref}
\hypersetup{hidelinks}

\newtheorem{proposition}{Proposition}

\def\phi{\varphi}

\def\({\left(}
\def\){\right)}

\def\b0{{\mathbf{0}}}

\newtheorem{Lemma}{Lemma}

\newtheorem{Remark}{Remark}

\title{Efficient Multiuser AI Downloading via Reusable Knowledge Broadcasting}
\author{
    Hai~Wu, Qunsong~Zeng, and Kaibin~Huang
    \thanks{
        The authors are with the Department of Electrical and Electronic Engineering, The University of Hong Kong, Hong Kong (Email: \{wuhai, qszeng, huangkb\}@eee.hku.hk). 
        Corresponding author: K. Huang.
    }
}
\makeatletter
\newcommand{\removelatexerror}{\let\@latex@error\@gobble}
\makeatother
\IEEEoverridecommandlockouts

\begin{document}
\maketitle
\begin{abstract}
For the \emph{sixth-generation} (6G) mobile networks, in-situ model downloading has emerged as an important use case to enable real-time adaptive \emph{artificial intelligence} (AI) on edge devices. 
However, the simultaneous downloading of diverse and high-dimensional models to multiple devices over wireless links presents a significant communication bottleneck.
To overcome the bottleneck, we propose the framework of \emph{model broadcasting and assembling} (MBA), which represents the first attempt on leveraging \emph{reusable knowledge}, referring to shared parameters among tasks/models, to enable parameter broadcasting to reduce communication overhead or latency.
The MBA framework comprises two key components. 
The first, the MBA protocol, defines the system operations including parameter selection from an AI library, power control for broadcasting, and model assembling at devices.
The protocol features the use of \emph{Shapley value} as a metric for measuring parameters' reusability. 
The second component is the joint design of \emph{parameter-selection-and-power-control} (PS-PC), which provides guarantees on devices' model performance and aims to minimize the downloading latency.
The corresponding optimization problem is simplified by decomposition into the sequential PS and PC sub-problems without compromising its optimality. 
The PS sub-problem is solved efficiently by designing two efficient algorithms. 
On one hand, the low-complexity algorithm of greedy parameter selection features the construction of task-oriented candidate model sets and a greedy selection metric for choosing the sets of model blocks for broadcasting, both of which are designed under the criterion of maximum reusable knowledge among tasks. 
On the other hand, the optimal tree-search algorithm gains its efficiency via the proposed construction of a compact binary tree pruned using model architecture constraints and an intelligent branch-and-bound search on the tree that fathoms nodes via solving a linear program that integer-relaxes the PS sub-problem. 
Last, given optimal PS, the optimal PC policy is derived in closed form by transforming the PC sub-problem into the conventional problem of energy-efficient transmission. 
Through extensive experiments conducted on real-world datasets, our results demonstrate the substantial reduction in downloading latency achieved by the proposed MBA design compared to traditional unicasting-based model downloading. 
\end{abstract}


\section{Introduction}
\label{sec: introduction}

Compared with past generations, a distinctive feature of the \emph{sixth-generation} (6G) mobile network is the expected ubiquitous deployment of machine learning and \emph{artificial intelligence} (AI) algorithms at the network edge \cite{Whitepaper_Huawei}.
In the two resultant active areas, \emph{edge learning} concerns the deployment of distributed learning algorithms to distill intelligence from mobile data while \emph{edge AI} focuses on the delivery of inference services to mobile devices \cite{6G_Letaief, Dis_learning_Chen}.
There exist two approaches for edge AI, one, termed \emph{split inference}, involves uploading of features extracted from raw data at devices to servers for remote inference using large-scale models \cite{Split_inf_Shao}.
The other approach, termed \emph{in-situ model downloading} and considered in our work, is to adapt AI-model downloading for on-device inference to users' current situation needs (e.g., context and location) \cite{Model_downloading_KB}.
Both approaches are being considered for standardization \cite{3GPP}.
Given limited radio resources, the implementation of edge AI is stymied by a communication bottleneck resulting from the high dimensionality of features and models transmitted over-the-air.
In this work, we propose a novel framework, called \emph{model broadcasting and assembling} (MBA), that exploits reusable knowledge, referring to the sharing of model parameters among heterogeneous tasks to reduce the communication overhead of in-situ model downloading.

Split inference obtains its name from the fact that a pair of device and server models are jointly trained such that they can be considered as sub-models resulting from splitting a global model \cite{Split_inf_Shao, JSCC_Deniz, Device_edge_Zhou}.
Deployed on the mobile-edge-computing platform, split inference provides mobile devices access to large-scale models in an edge cloud.
This endows devices with powerful AI capabilities that transcend their limitation in computing resources and energy.
To overcome the said communication bottleneck as well as accommodate the heterogeneity of mobile hardware, researchers have designed algorithms to optimally adapt the model-splitting point to available on-device computation resources, channel states, bandwidth, and latency requirements \cite{Device_edge_Zhou, Edge_AI_Li, Cooper_inference_Shao, Bottleneck++_Shao}.
The optimization is shown to yield significant system improvements in terms of end-to-end latency, accuracy, or throughput \cite{Wireless_AI_Park, PFTX_Lan, Batch_Liu}.
Despite these research efforts, the communication bottleneck is difficult to eliminate in scenarios with data-intensive applications and many devices, such as camera surveillance networks and collaborative auto-pilots.
Furthermore, potential information leakage due to feature uploading exposes users to security threats.

In-situ model downloading can address these issues as local inference is secure, and communication overhead is not excessive if task switching is infrequent and models of small-to-medium sizes suffice \cite{Model_downloading_KB}.
While the technology leverages increasingly powerful mobile AI processors, its practicality can be further improved using fast-advancing techniques for model pruning and quantization \cite{Compression_NN_Li}, and early exiting \cite{Batch_Liu}.
Though in-situ model downloading is embraced in the industry as a key edge-AI technology, relevant research is still in its nascent stage \cite{3GPP, Model_downloading_KB}.
In the largely uncharted area, one remotely relevant study is conducted on dynamically caching models at edge servers via downloading from the central cloud in response to incoming inference task requests from users \cite{Model_placement_Yan}.
In our view, a key challenge for in-situ model downloading, namely its communication bottleneck, arises in the scenario where many devices execute heterogeneous tasks and trigger the simultaneous downloading of a large number of different models.

We propose to tackle this challenge by exploiting the latest advancements in the area of \emph{knowledge reuse} that refers to the use of trained AI models to perform a \emph{different but related} task \cite{Foundation_Model_2021, Pathway}.
Two main modalities in the area are transfer learning and model reassembling (also known as model stitching).
The essence of transfer learning is to set a pre-trained model targeting one task as the initial point or supervisor for further training of another model to support other tasks \cite{TL_survey}.
Though broadcasting an identical model to devices for local transfer learning is communication efficient, retraining is a computation-intensive and energy-draining process that is usually unaffordable for devices.
Being the state-of-the-art of knowledge/model reuse, model reassembly \cite{Reassembly_2022} also has the potential of being developed into a practical solution for communication-efficient in-situ model downloading.
The technology generates a model for a task by reassembling parametric blocks borrowed from multiple pre-trained models targeting different tasks.
Hence, such a methodology is considered to be a special variant of transfer learning since the usage of multiple well-trained models on shelves.
As an option to enhance the model performance, extra ``stitching'' layers can be inserted in-between the borrowed blocks and tuned by few-shot training, which is much simpler than full-model training in transfer learning \cite{SNNet_2023, TL_survey}.
One practical landing of model reusing is the \emph{Pathway model} provided by Google \cite{Pathway}.
Multiple tasks are executed via traversing along distinct paths across the model that comprises a hierarchy of connected parametric blocks \cite{muNet_2022}.
The innovation of the model design lies in its compactness via extensive block sharing among different task paths.

In this work, we make the first attempt on exploiting knowledge reuse to enable communication-efficient in-situ model downloading for multiple devices executing heterogeneous tasks. 
The key feature of the proposed MBA framework is to replace the traditional unicasting of task-oriented models with the broadcasting of parametric blocks that help different tasks simultaneously.
Thereby, the downlink spectrum is better utilized and the downloading latency is reduced.
The main contributions of this work are the design of the MBA protocol that defines system operations and the knowledge-aware power control algorithm for joint \emph{parameter selection and power control} (PS-PC) for minimizing downloading latency. 
These two components of the MBA framework are summarized as follows.

\begin{itemize}
    \item \textbf{MBA Protocol.}
    The proposed MBA protocol consists of the following three steps. 
    First, a set of devices executing heterogeneous tasks send downloading requirements to an edge server together with their \emph{Quality-of-Service} (QoS) requirements in terms of task-specific model performance.
    Second, the server retrieves a set of parametric blocks from an AI library for sequential broadcasting to devices under their QoS requirements.
    Following \cite{Model_downloading_KB}, the library comprises a large number of blocks obtained by partitioning a broad range of pre-trained models.
    We propose to measure the reusability of each block for a given task using the well-known \emph{Shapley value} \cite{Shapley_Game_Theory_2002}.
    Third, in each round for transmitting a single block, the server sets the transmission power that determines the set of \emph{requesting} devices that can successfully receive the block over fading channels.
    Last, upon the completion of downloading, each device assembles a model optimized for the local task by selecting a subset from the received blocks.
    \item \textbf{Joint PS-PC.}
    The two server operations in the MBA protocol are coupled.
    On the one hand, parameter (block) selection aims to maximize the reusability of broadcast model blocks among requesting devices. On the other hand, the set of requesting devices is determined by power control.
    Exploiting the coupling, we consider the joint optimization of PS-PC to minimize the total downloading latency given an energy budget. It is proved that the complex problem can be simplified by decomposition into the sequential PS and PC sub-problems without compromising its optimality. The sub-problems are solved separately to yield a set of MBA algorithms as described below. 
    \begin{enumerate}
        \item \emph{Greedy parameter selection:} The design of the low-complexity algorithm, which approximately solves the PS sub-problem, has two main features: a) the construction of candidate model sets customized for the tasks of served devices and b) a greedy-selection metric that builds on blocks' reusability scores and is used for choosing model blocks from the candidate sets for broadcasting. Both the construction method and metric are designed under the same criteria of maximum reusable knowledge among tasks. 
        \item \emph{Optimal parameter selection:} The algorithm, which optimally solves the mixed-integer PS sub-problem via a tree search, gains its efficiency via the proposed construction of a compact binary tree aggressive pruned using model architecture constraints in the MBA protocol. Then an intelligent \emph{branch-and-bound} (B\&B) search on the tree is employed that fathoms nodes via solving a linear program that integer-relaxes the PS sub-problem. 
        \item \emph{Optimal power control:} Given parameter selection, the optimal PC policy is derived in closed form by transforming the PC sub-problem into the conventional problem of energy-efficient transmission. 
    \end{enumerate}
\end{itemize}

Experiments on real datasets demonstrate nearly 50\% latency reduction over 20 real-world image recognition AI tasks achieved by the optimized MBA framework as opposed to the conventional unicasting design without reusing knowledge among tasks.

The rest of this paper is organized as follows. 
In Section~\ref{sec: models and metrics}, models and metrics are introduced. 
The proposed MBA framework is presented in Section~\ref{sec: overview of model broadcasting and assembling}.
Then the problem of joint PS-PC is formulated and analyzed in Section~\ref{sec: PSPC}. 
The PS algorithms are designed in Section~\ref{sec: optimal parameter selection} and 
the optimal PC algorithm is presented in Section~\ref{sec: optimal power control}.
Experimental results are provided in Section~\ref{sec: experiments}, followed by concluding remarks in Section~\ref{sec: conclusions}.

\section{Models and Metrics}
\label{sec: models and metrics}

\begin{figure*}[t]
  \centering
  \vspace{3mm}
  \includegraphics[width=0.82\textwidth]{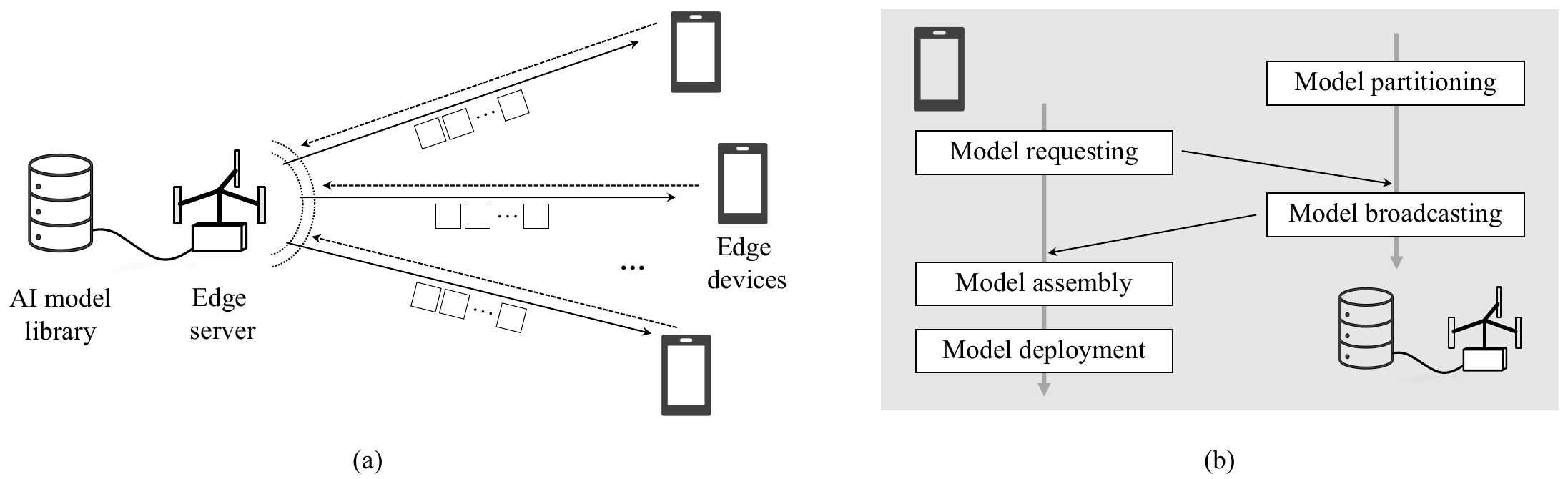}
  \caption{(a) Multi-user AI downloading system. (b) MBA operations and protocol.}
  \label{fig: framework}
  \centering
\end{figure*}

Consider a single-cell system where an edge server provides the AI downloading service to $K$ edge devices, as illustrated in Fig.~\ref{fig: framework}(a).
Their indices are denoted as $\mathcal{K} = \{1, 2, \cdots, K\}$. 
The devices aim to use downloaded models to accomplish heterogeneous AI tasks, e.g., object recognition and video rendering for augmented reality. 
Relevant models and metrics are described as follows.

\subsection{Block Cascading Model Architecture}
\label{subsec: deep neural networks}

\begin{figure*}[t]
  \centering
  \vspace{2mm}
  \includegraphics[width=0.82\textwidth]{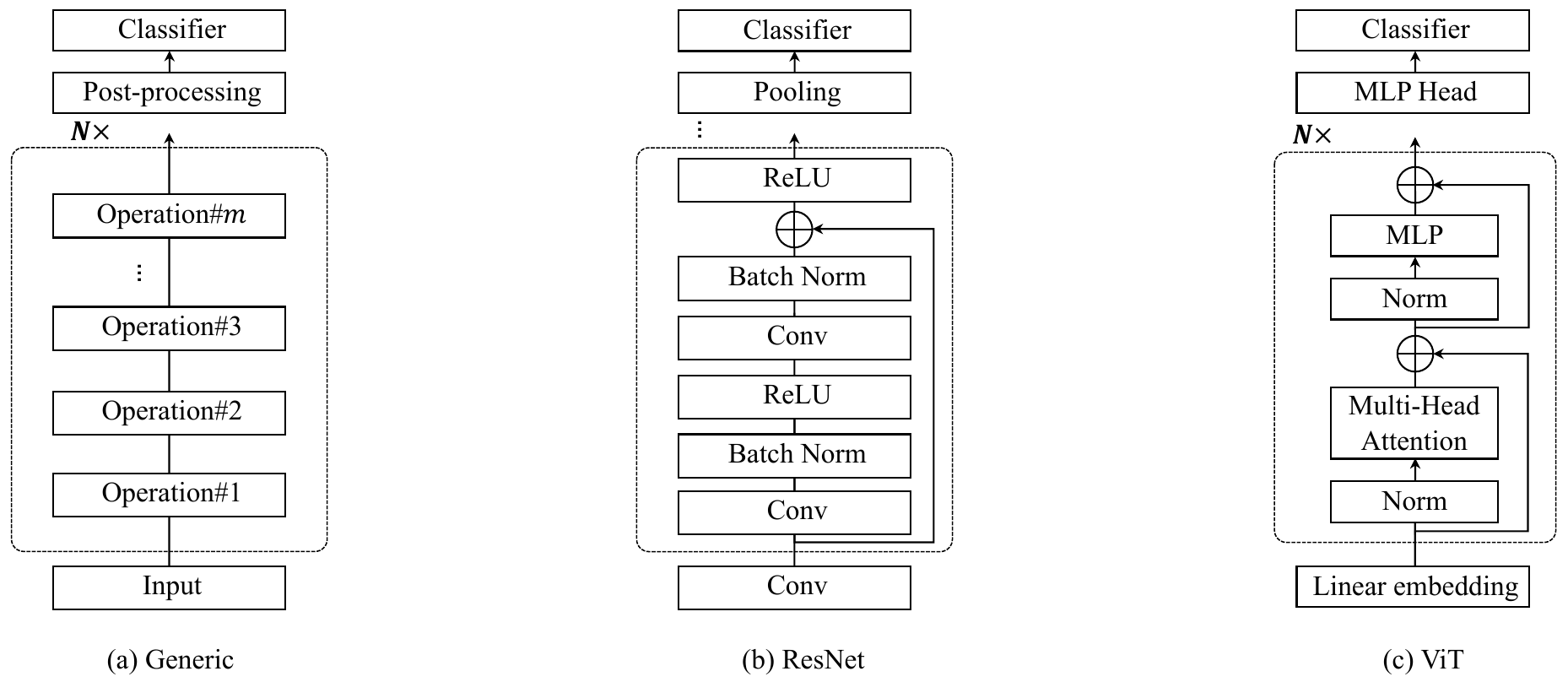}
  \caption{Three examples of block cascading architecture for DNNs, where a model block is framed in a dashed box. (a) The generic block cascading DNN architecture. (b) The neural network architecture of ResNet \cite{ResNet_2016}. Conv denotes a convolutional layer, and model blocks can be appended multiple times according to the different model sizes (i.e., number of layers). (c) The neural network architecture of ViT \cite{ViT_2021}. The transitional model block is repeated $N$ times based on different size settings (base/large/huge). Norm denotes layer normalization and MLP is a multi-layer perceptron.}
  \label{fig: DNN block}
  \centering
\end{figure*}

The state-of-the-art AI models are built on \textit{deep neural networks} (DNNs).
As illustrated in Fig.~\ref{fig: DNN block}, they are characterized by cascaded layers with each layer of neurons executing a specific operation on the input data. 
Then, the feedforward connections propagate the output of one layer to the next layer.
Consequently, the layerwise recursive operations decompose a network architecture into a series of cascaded blocks, each of which combines multiple operations as illustrated in Fig.~\ref{fig: DNN block}(a).
The detailed architecture of two widely-used DNN models, the ResNet that is a \textit{convolutional neural networks} (CNN) and the \textit{vision transformers} (ViT) are presented as examples in Fig.~\ref{fig: DNN block}(b) and (c) respectively. 
Excluding a dedicated input layer for feature sampling and a specialized final fully connected layer for classification or regression, the operations are typically implemented using cascaded operation blocks.
For a CNN, the operations within a block consist of sequential layers for convolution, normalization, activation, and residual addition.
Similarly, a ViT block consists of multi-head attention, normalization, residual addition, and multi-layer perceptions.
Thus, the operations on input data $\mathbf{x}$ can be represented by composite block operations, i.e.,
\begin{align}
    f(\{\mathbf{b}_n\}; \mathbf{x}) = \underbrace{f_{N}(\mathbf{b}_{N}; f_{N-1}(\mathbf{b}_{N-1}; \cdots f_1(\mathbf{b}_1; \mathbf{x}) \cdots ))}_\text{nested for $N$ times},
\end{align}
where $f_{n}(\cdot)$ denotes the operations of the $n$-th block with parameter $\mathbf{b}_n$ and $N$ denotes the number of blocks in one DNN model.
The entire model parameter set is represented by $\{\mathbf{b}_n\}$.
The edge server aims to construct $K$ DNN models using the AI model library to serve $K$ edge devices with heterogeneous AI tasks.
As a basis for constructing knowledge-reuse models for MBA, task-specific models refer to those models that are optimized for individual tasks without regard to unintended tasks.
Let the task-specific model for device $k$ be denoted as $\mathcal{S}_k = \{s_{k, 1}, s_{k, 2}, \cdots, s_{k, N}\}$, where the block at position $n$ is denoted as $s_{k, n}$.
Finally, each model block is compressed into $Q$ bits.

\subsection{Communication Model} 
\label{subsec: communication models}

Consider an arbitrary device, say device $k$.
The instantaneous \emph{channel state information} (CSI) is assumed to be known at the server. 
The channel power gain of the link between the server and the device is denoted as $H_k$, which remains constant throughout the entire downloading duration. 
The edge server sequentially broadcasts selected model blocks from the AI library since they can be potentially reused by all devices.
The transmission power for broadcasting block $s_{i,n}$ is denoted by $p_{i,n}$, which is constant given fixed CSI during the broadcasting process. 
Let $B$ represent the downlink bandwidth. 
Then the link rate for device $k$ during the broadcasting of $s_{i,n}$ is given as
\begin{equation}
    R_{i,n}^{(k)} = B \log_2\left(1 + \frac{p_{i,n}H_k}{N_0}\right), \quad \forall k \in \mathcal{K},
\end{equation}
where $N_0$ is the white Gaussian noise power. 
The edge server need to adjust the broadcasting (encoding) rates, denoted as ${C}_{i,n}$, of individual model blocks, which depend on transmission power $\{p_{i,n}\}$, to control their broadcasting latency. 
Specifically, given parameter selection, a block, say $s_{i,n}$, to be delivered to device $k$ is indicated by $\alpha_{i,n}^{(k)} = 1$ or otherwise $\alpha_{i,n}^{(k)} = 0$.
There can be multiple devices requiring $s_{i,n}$, which are called \emph{requesting} devices for this block.
Then, the broadcasting rate, $C_{i,n}$, of $s_{i,n}$, is determined by the requesting device with the worst channel state, i.e., $C_{i,n} = \min_{k \in \{k^{\prime} \mid \alpha_{i,n}^{(k^{\prime})} = 1\} } R_{i,n}^{(k)}$. 
As a result, the model blocks set successfully received at device $k$, denoted by $\Tilde{\mathcal{S}}_{k}$, is given as $\Tilde{\mathcal{S}}_{k} = \cup_{i=1,n=1}^{K, N} \alpha_{i, n}^{(k)} s_{i,n}$.
Subsequently, a block subset of $\Tilde{\mathcal{S}}_{k}$, is used to assemble a model for task execution on device $k$.

\subsection{Performance Metrics}
\label{subsec: QoS Metrics}

\subsubsection{Model Broadcasting Latency}
\label{subsubsec: downloading latency}

The metric measures the total latency for broadcasting such that the selected model blocks are successfully delivered to their requesting devices.
With each block comprising $Q$ bit, the broadcasting of a specific block, say $s_{i,n}$, to the requesting devices requires the latency denoted as $T_{i,n}$ and given as 
\begin{equation}
    T_{i,n}= \max_{k} \frac{\alpha_{i,n}^{(k)} Q}{R_{i,n}^{(k)}},
\end{equation}
which is determined by the requesting device with the worst channel state. 
Then the total broadcasting latency can be written as 
\begin{equation}
    T = \sum_{i=1}^K \sum_{n=1}^N T_{i,n}.
\end{equation}

\subsubsection{On-device Model Performance}
\label{subsubsec: model score}

The prediction performance of a model assembled by a device can be measured using a conventional metric such as inference accuracy or uncertainty, which, however, are intractable. 
For tractability, we instead adopt a surrogate metric, termed \emph{reusability score}, which sums the utility scores (i.e., Shapley values as elaborated in the following section) of model blocks for the current task across the target model. 
The metric is a popular \emph{explainable AI} (XAI) tool advocated in the literature to ensure transparency and accountability of different parameter components across diverse AI systems \cite{XAI_2019}.

\section{Overview of Model Broadcasting and Assembling (MBA)}
\label{sec: overview of model broadcasting and assembling}

\begin{figure}[t]
  \centering
  \includegraphics[width=0.49\textwidth]{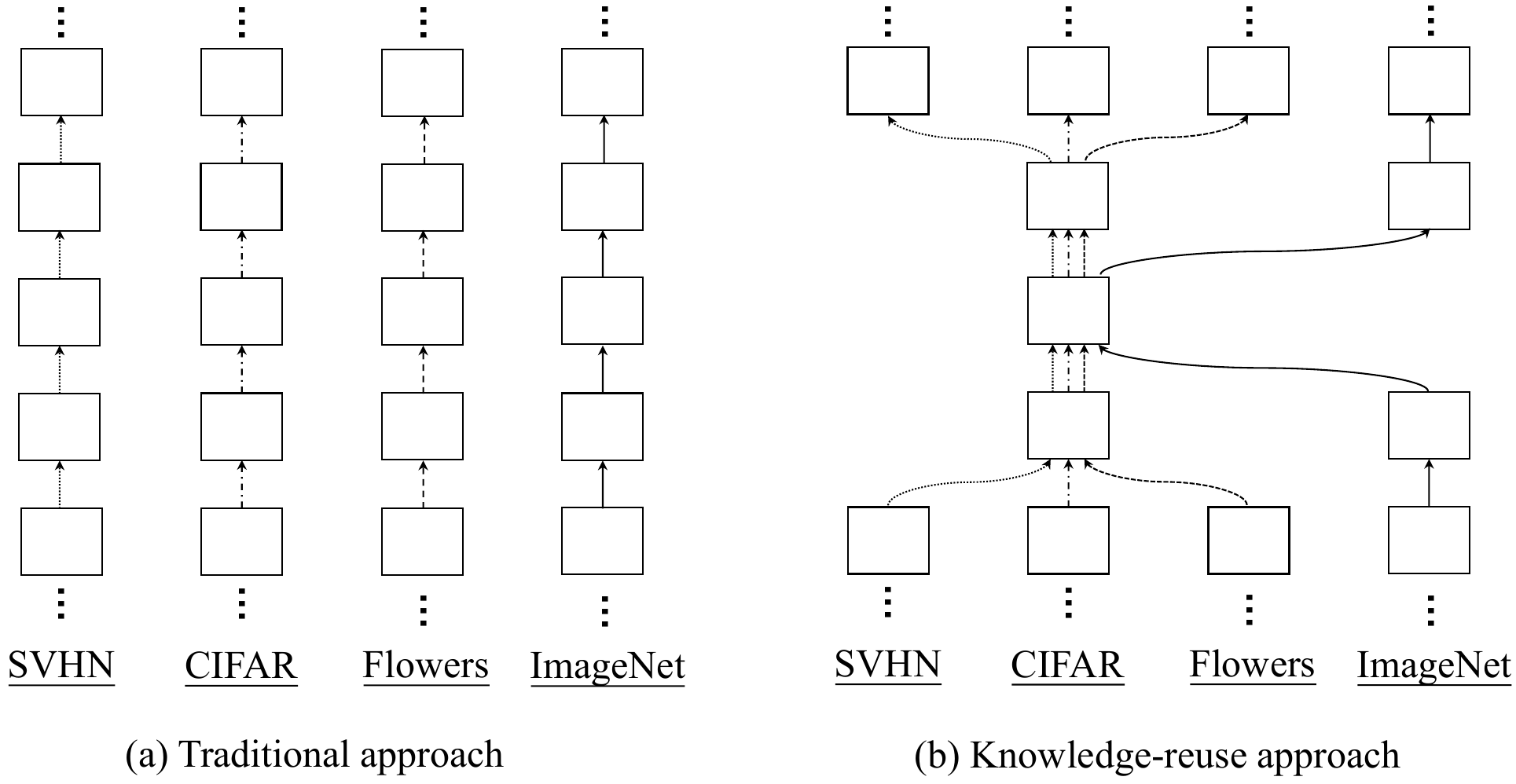}
  \caption{Model block forwarding process of AI models, where different image recognition tasks are distinguished by different directed lines. (a) Dedicated AI models in the traditional approach. (b) Multi-task AI models with shared model blocks in the knowledge-reuse approach.}
  \label{fig: block reuse}
  \centering
\end{figure}

In this section, we present the basic techniques and protocol in the MBA framework.
Underpinning the framework is the knowledge-reuse principle illustrated in Fig.~\ref{fig: block reuse}.
Therein, sharing common model blocks (i.e., reusable knowledge) enables the design of multi-task AI models more compact than the traditional dedicated models.
In the sub-sections, we design the basic techniques of model partitioning and parameter selection at the server and on-device model assembling.
Then building on the techniques, the MBA protocol is proposed.

\subsection{Enabling Techniques for MBA}

\subsubsection{AI Library Construction} 
\label{subsec: model partitioning and task-aware broadcasting}

The edge server is provisioned with an AI library constructed as follows.
First, consider a given large set of DNN models that as a whole can accommodate a wide range of tasks.
Each model is partitioned into $N$ parameter blocks.
Between every paired of models, their blocks at the same position $n$ in the model block sequences, i.e., $\mathcal{S}_k, \forall k \in \mathcal{K}$, can be swapped for knowledge reuse but changing any block position is prohibited \cite{Reassembly_2022}.
Given the possibility of block reuse via swapping, multiple knowledge-reuse models, which differ from the original task-specific model, can be found to meet the performance requirement of a specific task.
Targeting an arbitrary device executing a specific task, say device $k$, the finalized model used for on-device inference can be represented by a block-index sequence denoted as $\Tilde{\mathcal{S}}_k^{\star}$.
Given a set of task requests, the \emph{parameter-selection} algorithms, which are designed in the following sections, aim to find a set of knowledge-reuse models that can meet the task's QoS requirements and at the same time maximize their shared blocks to facilitating broadcasting.

Next, to enable parameter selection, we introduce a \emph{reusability score} to every block that measures its relevance to a given task.
The measure is based on the Shapley value which is a notion originally from game theory for accessing the credit of introducing a component in a coalition game and is now widely used to interpret the functionality of deep learning models \cite{Explain_Shapley_2019, Neuron_Shapley_2020, Shapley_Value_2020}.
In the current context, the Shapley value (reusability score) of model block $s_{i,n}$ for the task of device $k$ is equal to the average marginal contribution the block makes over all possible assembled models with performance assessed for the desired task of device $k$.
Let $I_{i, n}^{(k)}$ denote the reusability score of block $s_{i,n}$ \emph{with respect to} (w.r.t.) the task of device $k$.
Consider a knowledge-reuse model comprising the block sequence, e.g., $\mathcal{S}_i$.
Then the model's reusability score for device $k$ can be written as $I_i^{(k)} = \sum_{n=1}^{N} I_{i, n}^{(k)}$.

Prior to the deployment of the MBA service, the Shapley values of model blocks for disparate intelligent tasks, i.e., $\{I_{i,n}^{(k)}, \forall i, n, k\}$ are evaluated based on the algorithm proposed in \cite{Neuron_Shapley_2020} and recorded in the AI model library.
Based on the above discussion, we can readily define the mentioned AI library as an entity comprising: 1) a set of available model blocks; 2) a table of reusability score mapping the blocks to tasks; 3) a parameter-selection algorithm for assembling knowledge-reusable models (see Section~\ref{sec: optimal parameter selection}).

\vspace{1mm}
\subsubsection{Parameter Selection for Broadcasting} 
To leverage the reusability of model blocks in diverse AI tasks, an edge device is expected to receive as many model blocks as possible to maximize the possibility of assembling an ideal DNN model for its on-device inference.
Therefore, the system adopts a broadcast technique to transmit all the scheduled model blocks to all served devices tentatively.
However, with more model blocks needing to be broadcasted, the communication overhead, e.g., the downloading latency and energy consumption, increases.
Hence, parameter selection is in need to pick up the necessary model blocks, i.e., minimizing the number of broadcasted model blocks, to fulfill all served devices' QoS requirements.
The criterion behind parameter selection is to select blocks with high reusability scores for most AI tasks to minimize the total number of selected blocks.
Appended with a joint power control, the broadcasting minimizes the entire model download service latency.
The detailed joint PS-PC problem is formulated in Section~\ref{sec: PSPC}.
And the optimal PS and PC algorithms are elaborated in Section~\ref{sec: optimal parameter selection} and Section~\ref{sec: optimal power control} respectively.

\vspace{1mm}
\subsubsection{On-device Model Assembling}
\label{subsec: On-device Model Assembling}

Consider an arbitrary device, say device $k$.
Given the received model block set $\Tilde{\mathcal{S}}_{k}$, the model assembling selects part of model blocks to form $\Tilde{\mathcal{S}}_k^{\star}$, targeting the maximization of the reusability score while adhering to the model architecture constraint.
This constraint ensures that there is at least one available model block in $\Tilde{\mathcal{S}}_k$ for every position index $n$, which can be mathematically written as follows:
\begin{equation}\text{(C1: Model architecture constraints) \quad}
    \sum_{i=1}^K \alpha_{i,n}^{(k)} \geq 1, \forall n, k.
\end{equation}
\label{eq: model architecture constraint}
The device is then compelled to select the most valuable model block, i.e., the one with the highest reusability score concerning the desired task, from the received model block set $\Tilde{\mathcal{S}}_k$ for every position $n$, hence yielding $\Tilde{S}_k^{\star}$.
Consequently, the assembled model score of device $k$ is evaluated through
\begin{equation}
    C(\Tilde{\mathcal{S}}_{k}) = \sum_{n=1}^N \max\limits_{i}\{\alpha_{i,n}^{(k)} I_{i,n}^{(k)}\}.
\end{equation}

\subsection{MBA Protocol}

The steps of the MBA protocol as illustrated in Fig.~\ref{fig: framework}(b) are elaborated as follows. 
\begin{itemize}
    \item[1)] \textit{Model Requesting:} 
        Each device submits to the server a request for model downloading that specifies the corresponding AI task along with the QoS requirement.
    \item[2)] \textit{Parameter Selection and Broadcasting:} 
        Given the model downloading request, the parameter-selection algorithm is deployed to select a set of model blocks from the AI library for broadcasting.
        The transmission power is adapted using the power-control algorithm jointly designed with the parameter-selection algorithm in Sections~\ref{sec: optimal parameter selection}-\ref{sec: optimal power control}.
        The broadcasting maximizes the likelihood of assembling satisfactory DNN models at devices. 
    \item[3)] \textit{On-device Model Assembling:} 
        Upon the completion of broadcasting, each device assembles an AI model from the received block set to achieve the best performance on its AI task.  
\end{itemize}

\begin{Remark}(Is Knowledge Reuse a Free Lunch ?)
\label{remark: is knowledge reuse a free lunch}
\emph{The MBA design reduces downloading overhead by knowledge reuse among different downloaded AI models without compromising the task performance.
Careful readers may cast doubt on the seemingly ``free lunch" from knowledge reuse.
To address the concern, it should be emphasized that too aggressive knowledge reuse models inevitably lead to task performance degradation.
Therefore, underpinning the effectiveness of MBA is the optimal control of the ratio of overlapping blocks among downloaded AI models.
The range of knowledge-reuse level to achieve a ``free lunch" is further studied experimentally in Section~\ref{subsec: effectiveness of knowledge reuse}.}
\end{Remark}

\section{Joint Parameter-Selection-and-Power-Control: Problem Formulation and Decoupling}
\label{sec: PSPC}
\vspace{2mm}

In this section, the joint PS-PC problem for MBA is formulated with the objective of minimizing total downloading latency.
Then, we show that the problem can be decomposed into two sequential sub-problems, which simplifies the finding of optimal algorithms.

\subsection{Problem Formulation}
To begin with, two additional design constraints are introduced besides the model architecture constraints (C1) should be satisfied. 
First, the assembled model of an arbitrary device, say device $k$, should meet its QoS requirement, i.e., the reusability score on task $k$ exceeds some fixed threshold $c_k$, which are given in \eqref{eq: QoS requirements}.
\begin{equation}\text{(C2: QoS requirements) \quad}
\label{eq: QoS requirements}
    C(\Tilde{\mathcal{S}}_{k}) \geq c_k, \quad\forall k \in \mathcal{K}.
\end{equation}
The second is the energy constraint given the budget, $E$, for broadcasting.
\begin{equation}\text{(C3: Energy constraint) \quad}
    \sum_{i=1}^{K} \sum_{n=1}^{N} T_{i,n} p_{i,n} \leq E.
\end{equation}

Given the constraints (C1)-(C3) and the objective of broadcasting latency minimization, the problem of joint PS-PC can be formulated as
\begin{equation*}\text{(P1)~~}
    \begin{aligned}
        \min \limits_{\substack{\{p_{i,n}\}, \{T_{i,n}\}, \\ \{\alpha_{i,n}^{(k)}\}} }
        & \sum_{i=1}^{K} \sum_{n=1}^N T_{i,n}, \\
        \mathrm{s.t.~~~~~} 
        & \alpha_{i,n}^{(k)} Q \leq T_{i,n} B \log_2(1+ \frac{p_{i,n}H_k}{N_0}), \forall k, i, n, \\
        & \alpha_{i,n}^{(k)} \in \{0, 1\}, \forall i, n, k, \\
        & \text{(C1) \& (C2) \& (C3).}
    \end{aligned}
\end{equation*}

\vspace{1.5mm}
The combinatorial nature of parameter (i.e., model block) selection coupled with the problem of continued power control renders the above problem challenging to solve. 
Nevertheless, a tractable solution for problem (P1) can be found by decoupling the parameter selection (PS) and power control (PC) in the sequel.

\subsection{Equivalent Problem Decoupling}

It is found that the joint PS-PC problem in (P1) can be decomposed into two sequential sub-problems that decouple PS and PC.
First, the PS sub-problem is transformed from (P1).
The objective of the PS sub-problem is the number of selected blocks with relevant constraints extracted from problem (P1), which can be written as
\begin{equation*} \text{(P2: PS) \qquad}
    \begin{aligned}
        \min \limits_{\{\alpha_{i,n}^{(k)}\}} 
        & \quad \sum_{i=1}^{K} \sum_{n=1}^N \max_{k} \alpha_{i,n}^{(k)}, \\
        \mathrm{s.t.}~  & \quad \sum_{i=1}^K \alpha_{i,n}^{(k)} \geq 1, \forall n, k, \\
        & \quad \sum_{n=1}^N \max\limits_{i}\{\alpha_{i,n}^{(k)} I_{i,n}^{(k)}\} \geq c_k, \forall k, \qquad \qquad \qquad \quad \\
        & \quad \alpha_{i, n}^{(k)} \in \{0, 1\}, \forall i, n, k.
    \end{aligned}
\end{equation*}
Given the optimal model block indicator $\{{\alpha_{i,n}^{(k)}}^*\}$ from solving problem (P2), the PC sub-problem is reduced from (P1) as

\begin{equation*} \text{(P3: PC)~~} 
    \begin{aligned}
        \min \limits_{\substack{\{p_{i,n}\} \\ \{T_{i,n}\}}}
        &  \sum_{i=1}^{K} \sum_{n=1}^N T_{i,n} \\
        \mathrm{s.t.~}
        &  ~ {\alpha_{i,n}^{(k)}}^{*} Q \leq T_{i,n} {B \log_2(1+\frac{p_{i,n}H_k}{N_0})}, \forall k, i,n, \\
        &  ~ \sum_{i=1}^{K} \sum_{n=1}^{N} T_{i,n} p_{i,n} \leq E.
    \end{aligned}
\end{equation*}

\noindent The optimality of the above decoupling is shown in the following lemma.

\begin{Lemma}
\label{lemma: equivalent decoupling}
    \emph{Solving problem (P1) is equivalent to solving problem (P2) followed by problem (P3).
}
\end{Lemma}

\begin{proof}
First of all, the feasibility of problem (P1) implies the sufficiency of the energy budget.
Obviously, the total downloading latency and energy consumption are both monotone-increasing functions of the number of model blocks selected for broadcasting.
Then solving (P1) necessitates its minimization, which corresponds to problem (P2).
Conditioned on the resultant selected blocks, the optimal power control policy is reduced to the traditional problem of broadcasting latency minimization (see, e.g., \cite{power_control}), which corresponds to problem (P3).
This completes the proof.
\end{proof}


\section{Optimal Parameter Selection for MBA}
\label{sec: optimal parameter selection}

For the parameter selection, we propose two algorithms in this section, based on greedy selection and branch-and-bound (B\&B) method, to tackle mixed integer programming.

The $\max(\cdot)$ operation brought by the model assembling operation in problem (P2) makes the problem intractable. 
Therefore, we aim to simplify the reusability score constraints for the explicit analysis on model block parameter selection. 
Assuming the optimal solution towards problem (P2) is $\{{\alpha_{i,n}^{(k)}}^{*}\}$ and there exist multiple model blocks at position $n$ for device $k$ in the obtained solution, e.g., ${\alpha_{i_1, n}^{(k)}}^{\star} = {\alpha_{i_2, n}^{(k)}}^{\star} = \cdots = {\alpha_{i_k, n}^{(k)}}^{\star} = 1$. 
The solution indicates that the received model block exactly selected for model assembling, e.g., block $s_{i_1, n}$, included in $\Tilde{\mathcal{S}}_k^{\star}$, makes the assembled model fulfill the reusability score constraint (C2) as well as the architecture constraint (C1).
It is obvious that dropping the other model blocks at device $k$, i.e., setting ${\alpha_{i_2, n}^{(k)}}^* = \cdots = {\alpha_{i_k, n}^{(k)}}^* = 0$, will not increase the number of blocks for broadcasting and the reusability score constraints and the architecture constraint also hold.
Therefore, the $\max{} (\cdot)$ operations in model assembling can be discarded by compelling each device to reserve only one model block for each position during server parameter selection, which is sufficient to offer an optimal solution to the problem (P2).
Hence, problem (P2) can be transformed into the following problem (P2.1) equivalently, that is

\begin{equation*} \text{(P2.1)}
    \begin{aligned}
        \quad \min \limits_{\{\alpha_{i,n}^{(k)}\}, \{\alpha_{i,n}\} }
        & \quad \sum_{i=1}^{K} \sum_{n=1}^N  \alpha_{i,n}, \\
        \mathrm{s.t.~~} \quad
        & \quad \sum_{i=1}^K \alpha_{i,n}^{(k)} = 1, \forall n, k, \\
        & \quad \sum_{n=1}^N \sum_{i=1}^{K} \alpha_{i,n}^{(k)} I_{i,n}^{(k)} \geq c_k, \forall k, \qquad \qquad \qquad \quad \\
        & \quad \alpha_{i,n}^{(k)} \leq \alpha_{i,n}, \forall k, i, n, \\
        & \quad \alpha_{i, n}^{(k)} \in \{0, 1\}, \forall i, n, k, \\
        & \quad \alpha_{i, n} \in \{0, 1\}, \forall i, n.
    \end{aligned}
\end{equation*}

However, the simplified problem (P2.1) is a combinatorial optimization problem and remains complicated to solve as shown below.
\begin{proposition} 
    \emph{Problem (P2.1) is NP-complete.}
    \label{proposition: np-complete}
\end{proposition}
\noindent The proof is given in Appendix \ref{appendix: A}.
\vspace{0.3cm}

The traditional approach is to approximately solve problem (P2.1) using the integer relaxation method that, however, may not yield a feasible solution for the current case \cite{Relaxation}.
Specifically, if the said method is adopted, the binary block selection indicators $\{\alpha_{i,n}^{(k)}\}$ and $\{\alpha_{i,n}\}$ can be relaxed as $\alpha_{i,n}^{(k)} \in [0, 1]$ and $\alpha_{i,n} \in [0, 1]$.
Then problem (P2.1) is reduced to a linear programming optimization, which can be solved efficiently by the existing optimizers, e.g., \cite{cvx}.
After obtaining the indicators $\alpha_{i,n}^{(k)} \in [0, 1]$, the results are rounded to binary integers $\{0, 1\}$.
Though linear relaxation can yield a feasible solution to the relaxed problem, the rounding operation may make the block selection indicators unable to satisfy model architecture and reusability score constraints. 
To address this issue, we propose a new solution approach as presented in the sub-sections.

\vspace{-2mm}
\subsection{Greedy Parameter Selection}
\label{subsec: greedy ps}

The objective of problem (P2.1) is to minimize the cardinality of the union of the model block sets received by all devices. 
This requires finding the combination of model block sets for individual tasks to have the highest overlapping ratios, which is NP-complete as proved in Proposition~\ref{proposition: np-complete}. Alternatively, we propose a low-complexity scheme of greedy parameter selection that can achieve close-to-optimal performance. 
The scheme, which is summarized in Algorithm \ref{algorithm: greedy model parameter selection}, consists of two steps that are elaborated as follows.

\subsubsection{Construction of Candidate Model Sets} 
A dedicated candidate model set is generated for each task/device. 
Then a model request can be met by choosing from the associated candidate model set a model that not only satisfies the corresponding QoS requirement in (C2) but also supplies model blocks useful for other tasks. 
Therefore, the $K$ candidate model sets should share as many blocks as possible so as to maximize the effectiveness of broadcasting. 
To acquire this property, we design a construction method following a similar procedure as the training of muNet, a highly-compact multi-task model developed by Google, where model blocks for one task are cloned for another one sequentially \cite{muNet_2022}.
Without loss of generality, consider the construction of the candidate model set for device/task $k$, which is denoted as $\mathcal{C}_k$. 
To begin with, the first element (model) of $\mathcal{C}_k$, denoted as $\mathcal{C}_{k, 1}$, is set as the task-specific model optimized for task $k$, i.e., $\mathcal{C}_{k,1} = \mathcal{S}_k = \{s_{k,1}, \cdots, s_{k, N}\}$. 
Using $\mathcal{C}_{k, 1}$ as a seed, other elements of $\mathcal{C}_k$ are constructed sequentially as follows. 
Consider the construction of $\mathcal{C}_{k, 2}$. 
It is derived from $\mathcal{C}_{k, 1}$ by replacing only \emph{one} block, the one with the highest reusability score w.r.t. task $k$. 
The replacement block is chosen from all blocks with the same position index constituting all other task-specific models: $\{\mathcal{S}_i \mid i \neq k\}$ using the criterion of highest reusability score w.r.t. task $k$. 
Next, consider the construction of an arbitrary element of $\mathcal{C}_{k}$, say $\mathcal{C}_{k, n^{\prime}}$. 
Similarly as the preceding case, $\mathcal{C}_{k, n^{\prime}}$ is derived from $\mathcal{C}_{k, n^{\prime}-1}$ by replacing one of its blocks, which is identical to the block in the seed $\mathcal{C}_{k, 1}$ that has the $n^{\prime}$-th highest reusability score w.r.t. task $k$. 
Following the above sequential procedure, we can construct the model set $\mathcal{C}_k = \{\mathcal{C}_{k,1} \cdots, \mathcal{C}_{k, |\mathcal{C}_k|} \}$. 
Upon the construction of each model, we check if it can satisfy the QoS requirement of task $k$. 
If the result is negative, the model is discarded and the construction procedure is terminated, determining the cardinality $|\mathcal{C}_k|$.
Note that $|\mathcal{C}_k|$ cannot be larger than $N + 1$.
Other candidate model sets can be constructed in the same way, yielding the total of $K$ candidate model sets $\mathcal{C}_1, \mathcal{C}_2, \cdots, \mathcal{C}_K$.

\begin{algorithm}[t] \caption{Greedy Model Parameter Selection}
\label{algorithm: greedy model parameter selection}
    \textbf{Input:} The task index set $ \mathcal{K} = \{1, 2, \cdots, K\} $, all model sequences $\{\mathcal{S}_{i}\}$ with their task-specific blockwise reusability score $\{I_{i,n}^{(k)}\}$, the QoS score requirements $\{c_k\}$; \\
    \textbf{Candidate listing:} \\
    $\mathcal{C}_1 = \cdots = \mathcal{C}_K = \emptyset$;\\
    \textbf{for} $k \in\{1, \cdots, K\}$ \textbf{do} \\
    \qquad $\mathcal{C}_k = \{ \mathcal{S}_{k} \} $, $\mathcal{C}_{k, 1} = \mathcal{S}_{k} $,\\ 
    \qquad Sort $\mathcal{C}_{k, 1}$ according to $\{I_{k,n}^{(k)}\}$ in descending order; \\
    \qquad \textbf{for} $n^{\prime} \in\{1, \cdots, K\}$ \textbf{do} \\
    \qquad \qquad \textbf{find} replacement block $\hat{s}_{k, n^{\prime}}$,  \\
    \qquad \qquad $\mathcal{C}_{k, n^{\prime}+1} = \mathcal{C}_{k, n^{\prime}} \backslash \{s_{k, n^{\prime}}\} \cup \{\hat{s}_{k, n^{\prime}}\}$; \\
    \qquad \qquad \textbf{if} $\mathcal{C}_{k, n^{\prime}}$ meets QoS $c_k$:\\
    \qquad \qquad \qquad $\mathcal{C}_k = \mathcal{C}_k \cup \{\mathcal{C}_{k, n^{\prime}}\}$; \\
    \qquad \qquad \textbf{else} \\
    \qquad \qquad \qquad \textbf{break}; \\
    \qquad \textbf{end for} \\
    \textbf{end for} \\
    \textbf{Model selection:} \\
    $\tilde{\mathcal{C}}_0 = \emptyset$; \\
    \textbf{for} $k \in\{1, \cdots, K\}$ \textbf{do} \\
    \qquad \textbf{find} $\mathcal{C}_{k, n^{\star}}$ based on  $\tilde{\mathcal{C}}_{k-1}$ using \eqref{eq: greedy selection metric}, \\
    \qquad $\tilde{\mathcal{C}}_{k} = \tilde{\mathcal{C}}_{k-1} \cup \mathcal{C}_{k,n^{\star}} $;\\
    \textbf{end for} \\
    \textbf{return} $\tilde{\mathcal{C}}_K$.
\end{algorithm}

\subsubsection{Greedy  Parameter Selection} 
Given the candidate model sets $\mathcal{C}_1, \cdots, \mathcal{C}_K$ and $K$ tasks, the server sequentially selects one model from each set to have $K$ selected models in total using a greedy-selection metric designed in the sequel. 
The criterion for designing the metric is to maximize the level of parametric overlapping between the selected models so as to minimize the number of broadcast blocks. 
This is equivalent to minimizing the cardinality of the union of the sets of blocks constituting the selected models. 
Consider the selection of the $k$-th model from $\mathcal{C}_k$. 
Let $\tilde{\mathcal{C}}_{k-1}$ denote the union of the sets of blocks constituting the model selected up to the $k$-th step, namely $\mathcal{C}_1, \cdots, \mathcal{C}_{k-1}$. 
Two useful metrics are defined as follows. 
The first metric, called \emph{distance} and denoted as $ d(\mathcal{C}_{k, n^{\prime}}, \tilde{\mathcal{C}}_{k-1})$, measures the number of new blocks w.r.t. those in $\tilde{\mathcal{C}}_{k-1}$ if $\mathcal{C}_{k,n^{\prime}}$ is selected. Mathematically, $d(\mathcal{C}_{k, n^{\prime}}, \tilde{\mathcal{C}}_{k-1}) = |\mathcal{C}_{k,n^{\prime}} \cup \tilde{\mathcal{C}}_{k-1} | - |\tilde{\mathcal{C}}_{k-1}|$. 
The second metric, called \emph{reappearance} and denoted as $u(\mathcal{C}_{k,n^{\prime}})$, reflects the overlapping level (in terms of identical blocks) between $\mathcal{C}_{k,n^{\prime}}$ and those in the remaining candidate model sets to be processed in subsequent steps. 
Mathematically, $u(\mathcal{C}_{k,n^{\prime}}) = \sum_{\ell={k+1}}^{K} \sum_{n^{\prime}=1}^{|\mathcal{C}_{\ell}|} |\mathcal{C}_{k,n^{\prime}} \cap \mathcal{C}_{\ell, n^{\prime}}|$. 
To meet the said design criterion, it is desirable to minimize the distance metric but maximize the reappearance metric. 
Therefore, the greedy-selection metric can be suitably defined as their ratio. 
Let $\mathcal{C}_{k, n^{\prime}_k}$ denote the model selected from the candidate set $\mathcal{C}_{k}$ with $n_k^{\prime}$ denoting its index. 
Then the $K$ selected models can be represented by the index sequence, $n_1^{\prime}, n_2^{\prime}, \cdots, n_K^{\prime}$. 
The scheme of greedy model selection is to compute the indices sequentially as follows: 
\begin{equation}
n^{\star} = \arg\min_{n^{\prime}}\ \frac{d(\mathcal{C}_{k, n^{\prime}}, \tilde{\mathcal{C}}_{k-1})}{u(\mathcal{C}_{k,n^{\prime}})}.    
\label{eq: greedy selection metric}
\end{equation}

\subsection{Optimal Parameter Selection}
We also design an efficient tree-search algorithm to optimally solve the mixed-integer problem in (P2.1). Its low complexity is achieved in two aspects. 

First, a compact binary tree is constructed as follows. 
Each node is associated with one of the binary variables $\{\alpha_{i,n}^{(k)}, \forall i, n, k \}$ which are defined earlier to be the successful delivery indicators of every model block. 
As a result, the tree stacks $K \times N \times K$ tiers associated with the equal number of variables ranging from $\alpha_{0,0}^{(0)}$ to $\alpha_{K, N}^{(K)}$ as illustrated in Fig.~\ref{fig: bb_tree}.
Furthermore, each node is also associated with a restricted, integer-relaxed version of problem (P2.1) w.r.t. to a node $N_t$, denoted as $\mathcal{P}_{N_t}$, which helps the efficient tree search in the sequel.
In $\mathcal{P}_{N_t}$, those binary variables associated with nodes in the current tier, namely the tier comprising the currently visited node $N_{t}$, as well as those in the tiers above are fixed while the variables in the tiers below are relaxed into continuous ones in $[0, 1]$. 
Recall that the architecture constraints in problem (P2.1) require that there should be only one model block in a specific position for each device.
Then we can enforce the constraints to aggressively prune tree nodes, resulting in a sparse tree accelerating the search. 
The pruning is illustrated in Fig.~\ref{fig: bb_tree}.

\begin{figure}[t]
  \centering
  \includegraphics[width=0.486\textwidth]{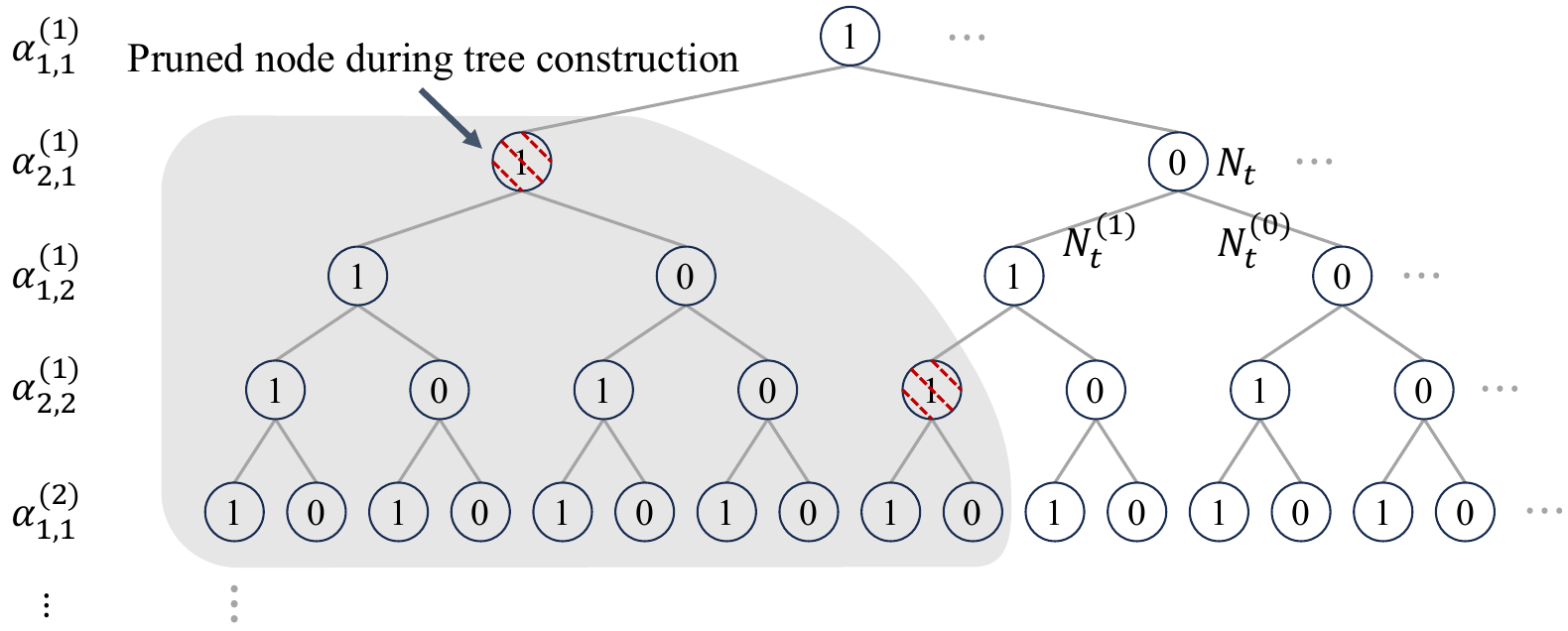}
  \caption{An example of iterative search tree for $K = N = 2$ used by Algorithm~\ref{algorithm: B&B model parameter selection}. Pruned nodes during tree construction that violate the model architecture constraints are boxed.}
  \label{fig: bb_tree}
  \centering
\end{figure}

Second, given the compact binary tree constructed above, we adopt the well-known B\&B method to design an algorithm for an efficient search over the tree. 
The most important operation in the algorithm is to fathom the current node $N_t$. 
This requires solving the problem $\mathcal{P}_{N_t}$ by using an existing highly efficient linear-program solver \cite{cvx}, for which the optimal objective (i.e., the minimized number of blocks) is denoted as $b_{N_t}^{\star}$. 
The fathoming has two possible outcomes. 
\begin{enumerate}
\item \emph{Node pruning:}
Node $N_t$ together with its children are pruned under one of the following two conditions. 
a) The devices' QoS requirements cannot be  all met as reflected in $\mathcal{P}_{N_t}$ being infeasible, so is the original problem (P2.1).
b) Let $b^{*}$ denote the evolving optimal objective of problem (P2.1). If $b_{N_t}^{\star} \geq b^{*}$, the communication overhead does not decrease by exploring node $N_t$.
c) The elements in the solution to $\mathcal{P}_{N_t}$ are binary, it is unnecessary to explore the child node of $N_t$ since it already provides the solution to the associated relaxed problems. 
\item \emph{Continuing the search:} If the conditions of the preceding outcome do not hold, the unvisited node list, denoted as $\mathcal{N}$, is updated by including two child nodes of $N_t$, $N_t^{(0)}$ and $N_t^{(1)}$ into $\mathcal{N}$. 
Then the search continues to fathom the next node in $\mathcal{N}$.
\end{enumerate}
Given the above fathoming operation, the designed B\&B algorithm can be summarized as follows. 
To begin with, the unvisited node list is initialized to have only the root node of the tree. 
Then the server pops one node from the list to fathom. 
The above procedure is repeated until the unvisited node list becomes empty. 
Upon the completion of the B\&B algorithm, the optimal solution to problem (P2.1) is found with the optimal objective of $b^{\star}$. 
The success of the optimal parameter selection is guaranteed since the search checks all the valid combinations of model blocks. 
The detailed steps of the algorithm are summarized in Algorithm \ref{algorithm: B&B model parameter selection}.

A naive design of B\&B algorithm could have the worst-case complexity of $\mathcal{O}(2^{KNK})$. 
The proposed construction of a compact tree reduces the worst-case complexity to $\mathcal{O}(K^{NK})$.
However, the reduced complexity can be prohibitive when the number of devices becomes large or the AI model backbone scales up. 
In such cases, the low-complexity scheme of greedy parameter selection in Section~\ref{subsec: greedy ps} provides a practical solution that can achieve comparable performance as the B\&B counterpart as shown by experiments in Section~\ref{subsec: parameter selection for MBA}.

\begin{algorithm}[t] \caption{Optimal Parameter Selection for MBA}
\label{algorithm: B&B model parameter selection}
    \textbf{Input:} The task index set $ \mathcal{K} = \{1, 2, \cdots, K\} $, all model sequences $\{\mathcal{S}_{i}\}$ with their task-specific blockwise reusability score $\{I_{i,n}^{(k)}\}$, the QoS score requirements $\{c_k\}$; \\
    \textbf{Initialize}: $t = 0$, \\
    \qquad Node list: $\mathcal{N} = \{N_1\}$, $b^{\star} = \infty$; \\
    \textbf{while} $\mathcal{N} \neq \emptyset$ \textbf{do} \\
    \qquad $t = t + 1$, \\
    \qquad \textbf{Node Visit:} select a node $N_t$ from $\mathcal{N}$, \\
    \qquad $\mathcal{N} = \mathcal{N}  \backslash \{N_t\}$,\\
    \qquad \textbf{if} model architecture constraints hold: \\
    \qquad \qquad \textbf{Node Assessment:} solve the node-associated \\ \qquad \qquad relaxed problem $\mathcal{P}_{N_t}$ and obtain solution \\ 
    \qquad \qquad with the objective of $b_{N_t}^{\star}$; \\    
    \qquad \qquad \textbf{Fathom Decision:} check the pruning criteria;\\
    \qquad \qquad \textbf{if} decide to preserve the node:\\
    \qquad \qquad \qquad $N_t^{(0)}$, $N_t^{(1)} \leftarrow $ two children of $N_t$,\\
    \qquad \qquad \qquad $\mathcal{N} = \mathcal{N} \cup \{N_t^{(0)}, N_t^{(1)}\}$; \\
    \qquad \qquad \textbf{end if} \\
    \qquad \qquad \textbf{if} pruning condition c) holds: \\
    \qquad \qquad \qquad $b^{*} = \min (b^{*}, b_{N_t}^{\star})$; \\
    \qquad \qquad \textbf{end if} \\
    \qquad \textbf{end if} \\
    \textbf{end while} \\
    \textbf{return} $b^*$ and associated solution
\end{algorithm}

\section{Optimal Power Control for MBA}
\label{sec: optimal power control}
In the preceding section, parameter selection for MBA is optimized by solving the sub-problem (P2). 
In this section, given optimal parameter selection, i.e., $\{{\alpha_{i,n}^{(k)}}^{\star}\}$, another sub-problem (P3) is solved to provide the optimal power-control algorithm for MBA. 

Our design approach is to leverage techniques in the rich, existing literature on power control by transforming problem (P3) to have a form similar to a traditional energy-constrained power allocation problem (see, e.g., \cite{power_control}). To this end, by solving problem (P2), the model blocks selected for a specific device, say device $k$, are specified by the indicators ${\alpha_{i,n}^{(k)}}^{\star} = 1, \forall i, n$. 
For an arbitrary model block, say $s_{i, n}$, the model block should be broadcast if there exists at least one device requesting it, which is denoted by $\alpha_{i,n}^{\star} = 1$ and otherwise $\alpha_{i,n}^{\star} = 0$.
Due to broadcasting, any other device, say device $k^{\prime}$, with better channel condition, i.e., $H_{k^{\prime}} \geq H_{k}$, can also receive this block.
Hence, the downloading latency for $s_{i,n}$ is determined by the requesting device with the worst channel state. 
Therefore, the equivalent channel state in transmitting model block $s_{i,n}$, denoted by $H_{i,n}$, is given as
\begin{equation}
    H_{i,n} = \begin{cases}
        \min \limits_{k \in \{k^{\prime} \mid \alpha_{i,n}^{(k^{\prime})} = 1\} } H_k, \quad & {\alpha_{i,n}}^{\star} = 1, \\
        0, & \text{otherwise}.
    \end{cases}
\end{equation}
Moreover, problem (P3) can be transformed as

\begin{equation*} \text{(P3.1)~~} 
    \begin{aligned}
        \min \limits_{\substack{\{T_{i,n}\}, \\ \{p_{i,n}\}}}
        & ~~ \sum_{i=1}^{K} \sum_{n=1}^N T_{i,n} \\
        \mathrm{s.t.~} 
        &  ~~ {\alpha_{i,n}}^* Q \leq T_{i,n} {B \log_2(1+\frac{p_{i,n}H_{i,n}}{N_0})},  \forall i,n, \\
        &  ~~ \sum_{i=1}^{K} \sum_{n=1}^{N} T_{i,n} p_{i,n} \leq E.
    \end{aligned}
\end{equation*}
The  equivalence between problem (P3) and problem (P3.1) can be proved by analyzing the \emph{Karush–Kuhn–Tucker} (KKT) conditions of problem (P3) after substituting the obtained $\{{\alpha_{i,n}^{(k)}}^{\star}\}$. 
Then, the optimal power control policy can be found by solving the convex problem (P3.1) using the conditions with the results shown in the following proposition. 
\begin{proposition}[Optimal power control]
    \emph{For MBA, the optimal transmit power $\{{p_{i,n}}^*\}$ that solves problem (P3.1) and the resultant downloading latency $\{{T_{i,n}}^*\}$ are given as}
    \begin{equation}
        \begin{aligned}
            & {p_{i,n}}^{*} = \frac{N_0}{H_{i,n}} \left(\exp{\left(W\left(\frac{H_{i,n}-\beta^{*}N_0}{\beta^{*}N_0 e}\right)+1\right)} - 1 \right), \\
            & {T_{i,n}}^{*} = \frac{Q\ln{2}}{B\left(W\left(\frac{H_{i,n} - \beta^{*}N_0}{\beta^{*}N_0e}\right) + 1\right)}, \\
        \end{aligned} 
    \label{eq: power control}
    \end{equation}
    \emph{where $W(\cdot)$ is the Lambert $W$ function (principal branch) and $\beta^*$ denotes the optimal Lagrange multiplier that solves the following equation:}
    
    \begin{equation}
        \frac{N_0 Q \ln{2}}{B} \sum_{i=1}^{K} \sum_{n=1}^N \frac{ \exp{\Big(W\Big(\frac{H_{i,n} - \beta^{*}N_0}{\beta^{*}N_0e}\Big) + 1\Big)} - 1}{ H_{i,n} \Big(W\Big(\frac{H_{i,n} - \beta^{*}N_0}{\beta^{*}N_0e}\Big) + 1 \Big)} = E.
    \label{eq: multiplier}
    \end{equation}

\end{proposition} 

\noindent The proof is provided in Appendix~\ref{appendix: B}.

\section{Experimental Results}
\label{sec: experiments}

\subsection{Experimental Setup}

\begin{enumerate}
\item \emph{AI models:}
In Section~\ref{subsec: effectiveness of knowledge reuse}, to assess the effectiveness of knowledge reuse among different AI tasks, we conduct experiments using the ImageNet-1K pre-trained ViT base model \cite{ImageNet_Russakovsky}. 
We fine-tune the model for image classification on the CIFAR-10 \cite{CIFAR_Krizhevsky} and SVHN \cite{SVHN} datasets. 
The images are resized to $224 \times 224$ to align with the ImageNet dataset, and patch embedding is performed.
We freeze a portion of the pre-trained model and sequentially fine-tune it by freezing additional two blocks from the head. 
The models are optimized using stochastic gradient descent with a momentum of $0.9$ and an initial learning rate of $0.03$.
A linear warmup and a linear learning rate scheduler are applied, and the training is performed for $5$ epochs with a fixed batch size of $256$.
In Section~\ref{subsec: parameter selection for MBA} - Section~\ref{subsec: inference performance of MBA}, to demonstrate the superiority of the proposed MBA framework, we deploy multi-task AI models with overlapping model blocks as obtained from the checkpoints released in \cite{muNet_2022}. 
The models are based on the ViT base backbone, with $16$ self-attention heads per transformer layer and $N = 24$ transitional blocks. 
In the model downloading service, the edge server provides $21$ different AI models for distinct image recognition tasks, selected from the multitask system described in \cite{muNet_2022}. 
The reusability score constraints on requested tasks are set to keep the on-device inference performance comparable with the dedicated trained models.
Each model block consists of $12,582,912$ parameters, which are compressed to $Q = 5$ Mbits after applying parameter pruning and quantization.
All AI models are fine-tuned using dual NVIDIA RTX-3090 GPUs with an AMD 5900X CPU.

\item \emph{Communication Settings:} 
The bandwidth for the broadcasting service is $B = 100$ MHz. 
The channel gains of the edge devices are assumed to follow i.i.d. Rayleigh fading with a path loss of $10^{-3}$. 
The noise spectral density is set to $N_0 = 0.5 \times 10^{-9}$ W/Hz. 
There are $|\mathcal{K}| = 21$ devices requesting AI models.
The default server energy budget for a single epoch of model downloading service is $E = 250$ J. 
The AI tasks requested by edge devices are randomly generated to be non-identical.

\item \emph{Benchmark Schemes:}
\begin{itemize}
    \item \textit{Constant-power MBA}: 
        This scheme is the proposed MBA but without power control.
        It serves the purpose of demonstrating the performance gain of joint PS-PC of the MBA framework.
        In other words, the server adopts constant transmit power.
    \item \textit{Channel-adaptive unicasting}: 
        The scheme represents the traditional approach where the server unicasts each requested AI model to the specific edge device. 
        Every device is required to download $N$ dedicated model blocks to meet the constraints of the model architecture and model score. 
        The server performs channel-adaptive transmit power control \cite{power_control} to minimize the expected duration of the entire model downloading service.
\end{itemize}
\end{enumerate}

\begin{figure}[t]
  \centering
  \includegraphics[width=0.438\textwidth]{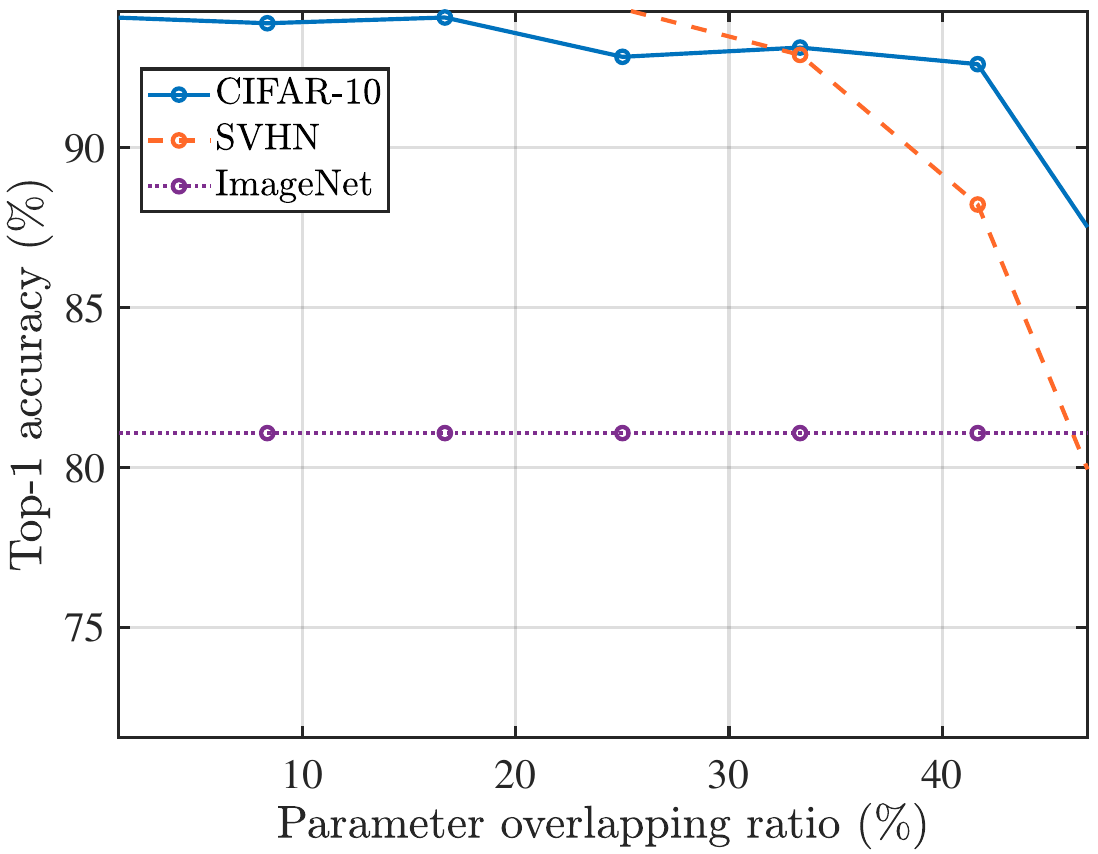}
  \caption{The effects of model parameter overlapping ratio of AI model pairs on the Top-1 inference accuracies of individual tasks as identified by their test datasets. 
  The overlapping ratio refers to that of two model pairs: 1) ImageNet and CIFAR-10; 2) ImageNet and SVHN.}
  \label{fig: acc_vs_O}
\end{figure}

\subsection{Effectiveness of Knowledge Reuse}
\label{subsec: effectiveness of knowledge reuse}

The relationship between the parameter-overlapping ratio and the Top-1 accuracy on the dataset for fine-tuning is depicted in Fig.~\ref{fig: acc_vs_O}.
The results demonstrate that if the overlapping ratio goes beyond a threshold (i.e., $30\%$), the task performance decreases both the Imagenet-CIFAR-10 and the ImageNet-SVHN pairs.
The performance loss is due to the increased number of model blocks that are frozen.
However, if the ratio is not too large (i.e., $0 - 30\%$), the performance degradation is negligible w.r.t. independently trained models (corresponding to the overlapping ratio being $0$).
This validates Remark~\ref{remark: is knowledge reuse a free lunch}.
The results show that compact multi-task AI models can be trained using advanced techniques \cite{muNet_2022, Reassembly_2022, SNNet_2023}.
They also support the principle of our design that knowledge reuse can be explored to reduce model downloading overhead without compromising task performance if the overlapping ratio is regulated.

\subsection{Parameter Selection for MBA}
\label{subsec: parameter selection for MBA}

\begin{figure}[t]
  \centering
  \includegraphics[width=0.45\textwidth]{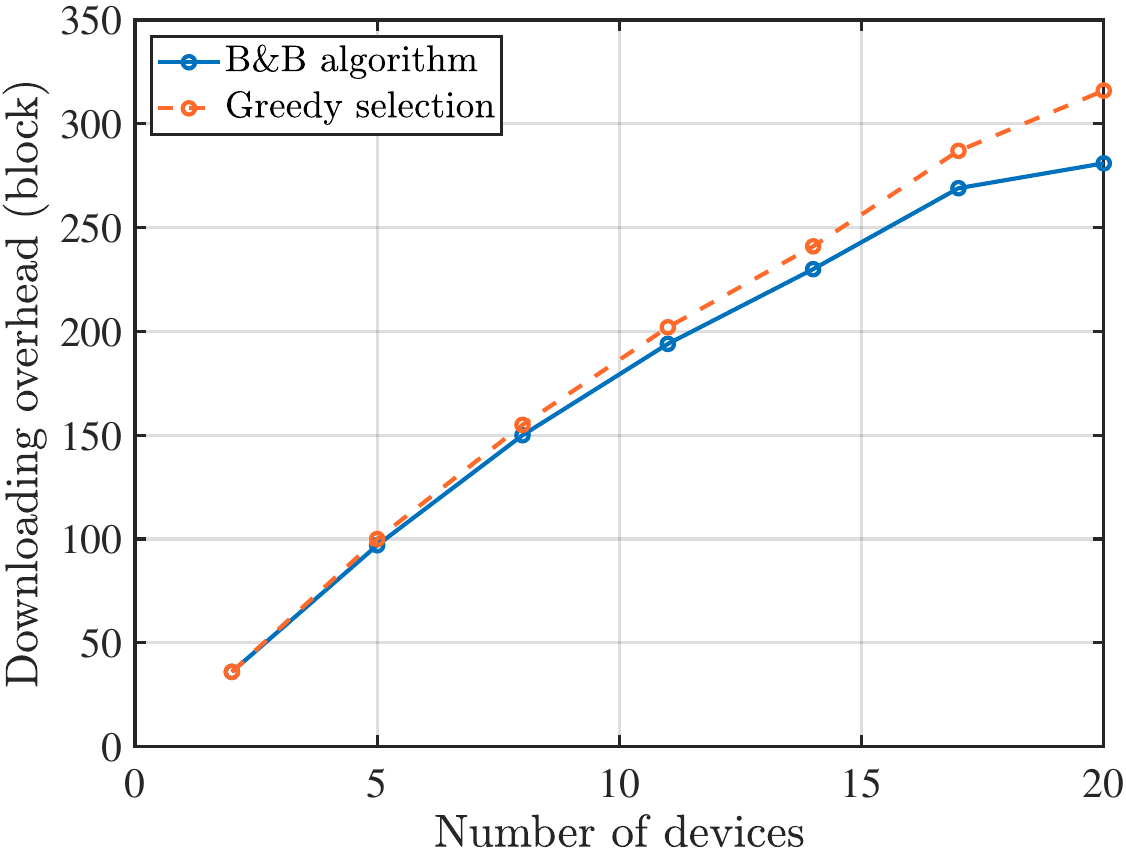}
  \caption{The comparison of downlink communication overhead for a varying number of devices between the greedy selection and branch-and-bound (B\&B) algorithm for joint PS-PC.}
  \label{fig: parameter selection}
\end{figure}

In this sub-section, we compare the performance of two joint PS-PC algorithms, namely the low-complexity algorithm of greedy parameter selection (see Algorithm \ref{algorithm: greedy model parameter selection}) and the B\&B method (see Algorithm \ref{algorithm: B&B model parameter selection}).
The metric is the number of model blocks selected for downloading to serve a given number of devices with reusability score requirements (see the objective of problem (P2.1)).
The metric quantifies the downloading overhead under users' QoS constraints.
In Fig.~\ref{fig: parameter selection}, we plot the curves of downloading overhead (i.e., the number of selected blocks for downloading) versus a varying number of devices.
One can observe that in the range of small-to-moderate number of devices ($1-8$), the greedy selection algorithm achieves identical performance as its optimal counterpart.
A performance gap starts to emerge when the number of devices exceeds this range.
The gap, however, can be observed to be small, e.g., a $5\%$ increase in downloading overhead at the number of devices equals $17$.
From the results, we can conclude that the greedy selection algorithm provides a low-complexity solution for joint PS-PC at the cost of small performance loss as opposed to the higher-complexity optimal solution.

\subsection{Latency Performance of MBA}

\begin{figure}[t]       
    \centering
    \includegraphics[width=0.44\textwidth]{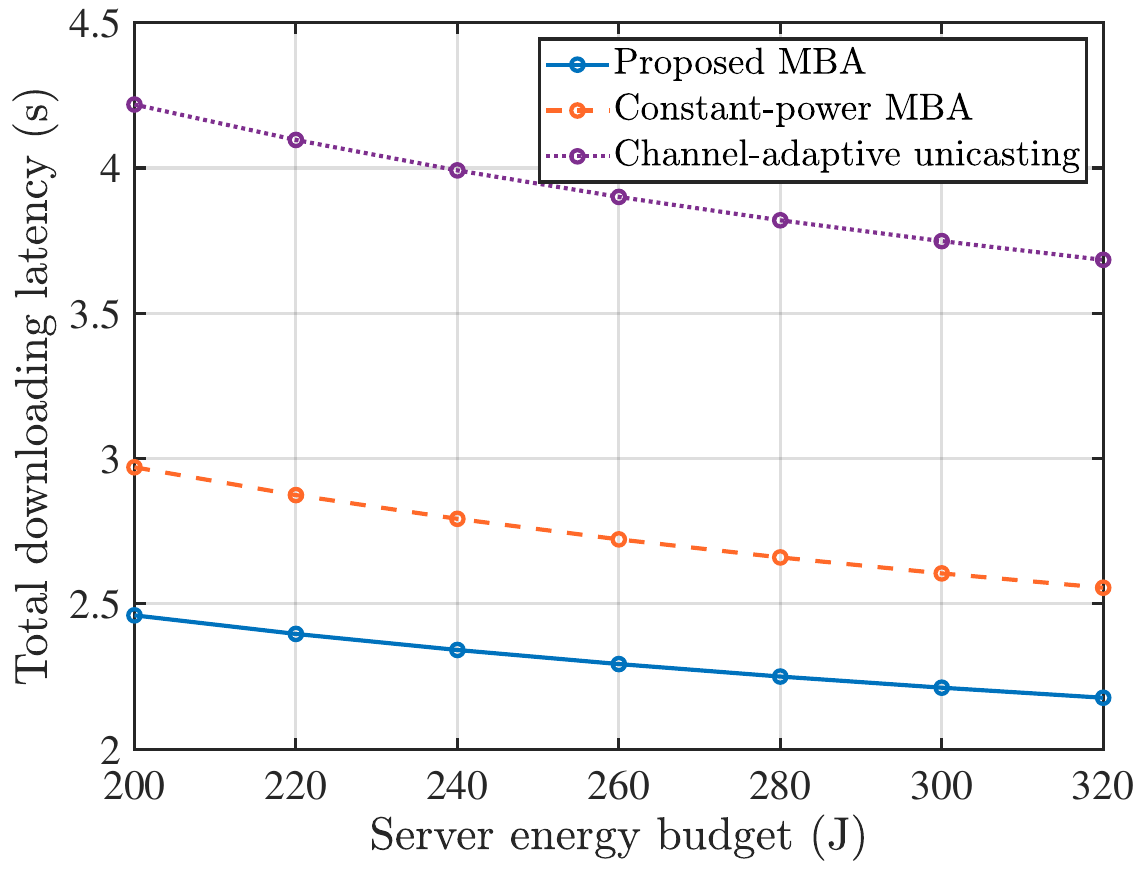}
    \caption{The effects of the server energy budget on the total downloading latency.}
    \label{fig: E_vs_T}
    \centering
\end{figure}

We consider the total downloading latency to meet devices' QoS requirements in terms of reusability score. 
Fig.~\ref{fig: E_vs_T} shows the curves of total downloading latency versus the server energy budget.
One can observe that MBA significantly reduces the latency with respect to benchmarking schemes.
In particular, MBA consistently achieves approximately $50\%$ lower latency than the traditional unicasting transmission.
On the other hand, it is found that unleashing the full potential of knowledge reusability requires dedicated power control, or otherwise, the level of latency reduction is substantially lower.
Power control is especially important for a low energy budget.
For instance, compared with MBA, its counterpart without power control increases the latency by $20\%$ and $10\%$ for the energy budget of $200$ J and $320$ J respectively.

Next, Fig.~\ref{fig: K_vs_T} shows the curves of downloading latency per device versus the number of devices. 
Generally, the latency increases as the number of users grows. 
It is observed that MBA achieves the minimum latency among all schemes throughout the considered range of the number of devices.
Moreover, its latency growth is gradual in contrast with benchmarking schemes.
As a result, their latency gaps are widened as the number of devices increases.
In particular, at $20$ devices, the latency of MBA is only $60\%$ of channel-adaptive unicasting and $85\%$ of MBA without power control.
It is worth emphasizing that the advantage of MBA is more pronounced when the number of devices is large.
The reason is the existence of more candidate parameter blocks for a particular position of a given device's model architecture, facilitating knowledge reuse and hence broadcasting.
The advantage endows on MBA scalability.

\begin{figure}[t]
    \centering
    \includegraphics[width=0.45\textwidth]{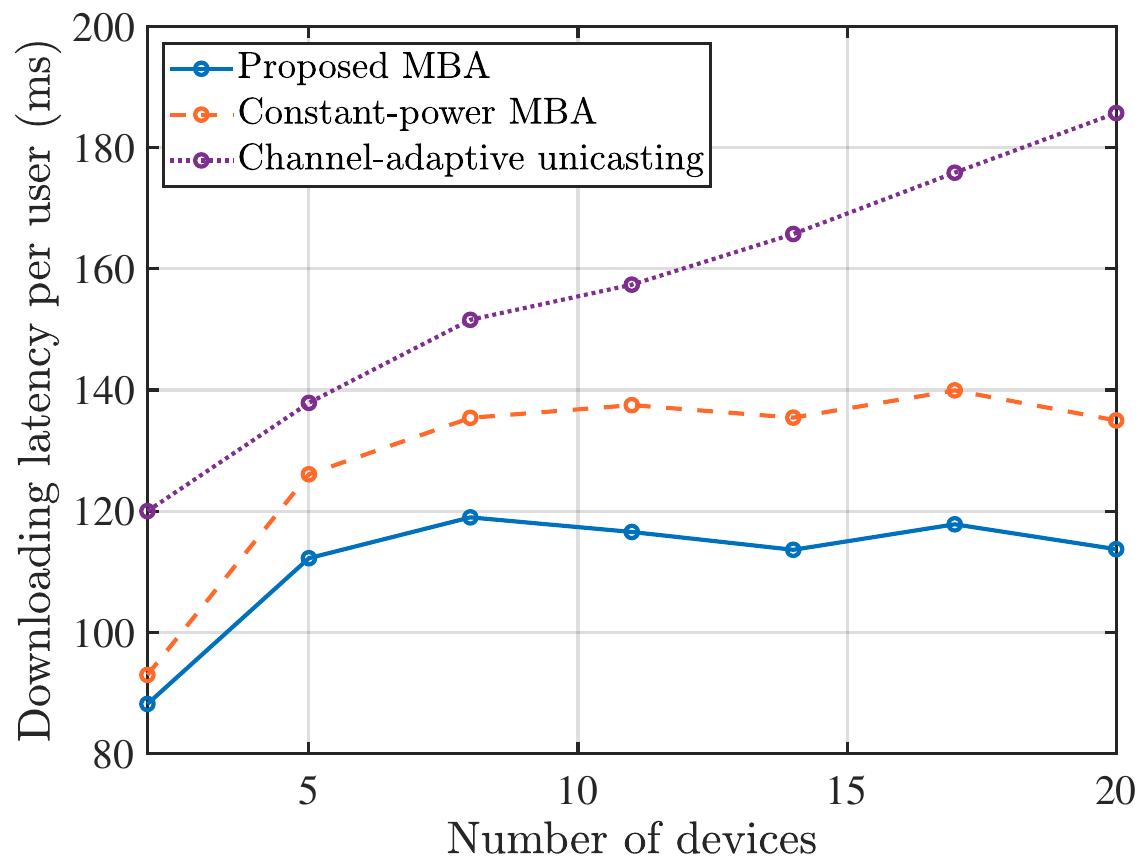}
    \caption{The effects of the number of devices on the total downloading latency.}
    \label{fig: K_vs_T}
\end{figure}

\subsection{Inference Performance of MBA} 
\label{subsec: inference performance of MBA}

In this sub-section, we consider the metric of on-device inference accuracy and a system with two devices executing two different tasks, namely classification on the CIFAR-10 and SVHN datasets.
The downloading latency requirements are assumed identical for the devices. 
We adopt the ViT base model for simulation simplicity, where the number of transitional blocks is $N=12$.  
Given a requirement, incomplete downloading may occur. 
To account for this case, each device has a backup model trained on the large dataset of ImageNet. 
Specifically, a device that fails to receive the complete desired model from the server can utilize model blocks from the backup model to assemble a model under the model architecture constraint at the cost of lower inference accuracy.

\begin{figure}[t]
    \centering
    \includegraphics[width=0.44\textwidth]{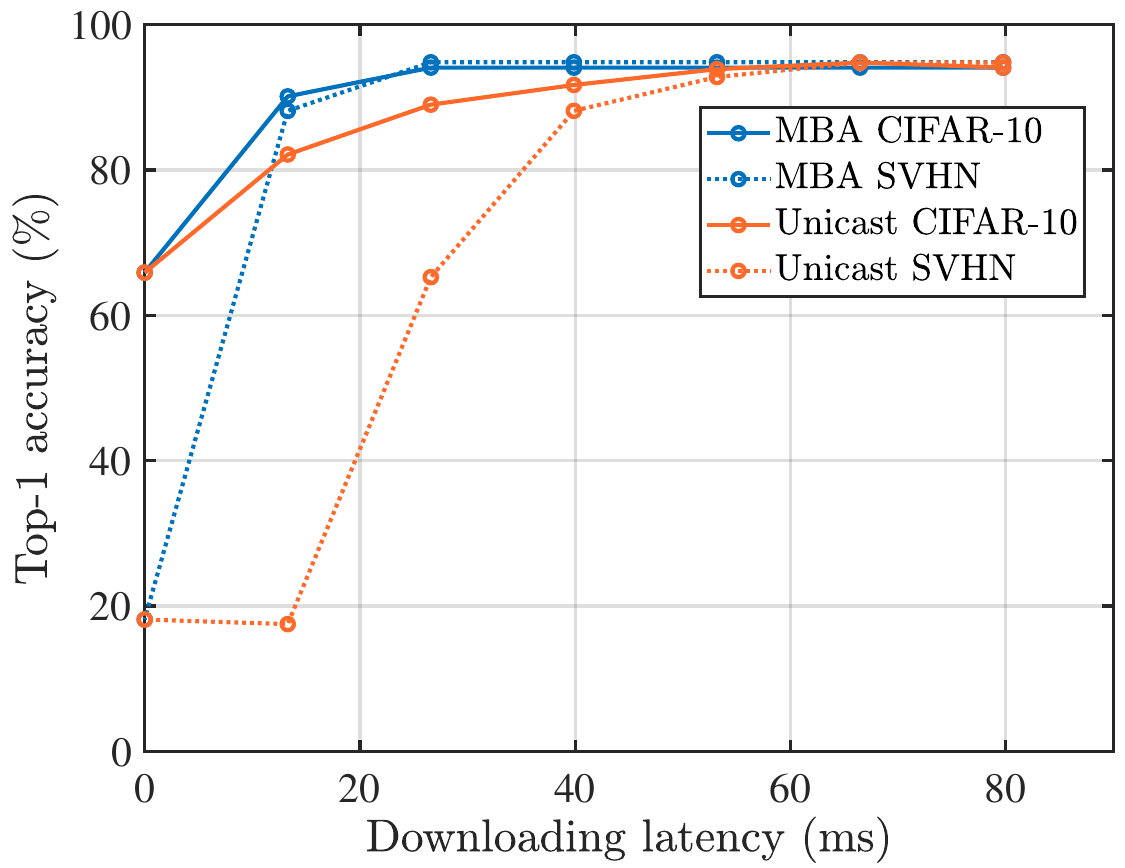}
    \caption{The comparison between the proposed MBA framework and the traditional unicast content delivery in terms of inference performance. Specifically, the curves show the effects of the given entire model downloading latency constraint on the Top-1 inference accuracy regarding respective datasets: one for CIFAR-10, and another for SVHN.}
    \label{fig: acc_vs_T}
\end{figure}

In Fig.~\ref{fig: acc_vs_T}, we plot the on-device Top-1 inference accuracy against the given downloading latency requirements.
It is observed that when devices do not download any model blocks from the edge server, i.e., latency equals $0$, the backup model lacking task-specific knowledge performs poorly, achieving only $18.14\%$ accuracy on the SVHN dataset and $65.89\%$ accuracy on the CIFAR-10 dataset. 
The increase in downloading latency leads to rapid performance improvements for all schemes and tasks.
The proposed MBA design shows a well-balanced performance of two tasks, both of which achieve close-to-maximum accuracies when the latency exceeds $20$ ms.
On the contrary, traditional unicasting causes imbalanced performance of the devices.
The scheme favors the CIFAR-10 task while the SVHN device suffers from a lower accuracy in the latency range of $0$ - $40$ ms.
The imbalance results from that the scheme tends to download more CIFAR-10 blocks so as to meet the total latency requirement.
The results demonstrate the superiority of MBA in situations with stringent latency constraints.

\section{Concluding Remarks}
\label{sec: conclusions}
In this work, we consider the 6G use case of AI downloading and propose that knowledge can be reused across different tasks to improve the efficiency of multiuser model downloading via broadcasting. 
The principle is implemented using the MBA framework.
In this framework, the usefulness of model blocks, which are stored in the cloud AI library, for different tasks, are measured by their reusability scores (i.e., Shapley values). 
Based on the scores, a minimum number of model blocks can be selected for broadcasting to meet the QoS requirements of devices with minimum transmission latency. 
Another finding in this work is the importance of task-oriented communications for unleashing the full potential of AI downloading. 
In particular, the power control for broadcasting is shown to significantly shorten the downloading latency if it is jointly designed with model block selection. 
Overall, by optimal control, MBA can realize a ``free lunch" for AI downloading via communication latency reduction without compromising multiuser task performance. 

The proposed exploitation of reusable knowledge for improving communication efficiencies opens a new direction for designing 6G task-oriented communication techniques to support distributed learning and inference at the edge. 
In the current context, it is interesting to further develop the MBA framework to incorporate more sophisticated physical-layer techniques such as broadband multi-access and broadcasting beamforming, to account for user fairness and highly skewed popularity of tasks, to support split learning or fine-tuning of large-scale foundation models (e.g., GPT) to support a wide range of downstream tasks. 
At a high level, the current work contributes to the ongoing paradigm shift from traditional content delivery networks towards AI-centric networks featuring ubiquitous AI model caching/downloading. 

\appendix

\subsection{Proof of Proposition 1}
\label{appendix: A}

To prove that problem (P2.1) is NP-complete, we can show that it is a generalization of the multiple choice \emph{max-min knapsack} (MNK) problem, which is known to be NP-complete \cite{Max-Min}.
The MNK problem is to fill the knapsack with selected items so that the minimum summed value under all scenarios is maximized.
Given problem (P2.1), the goal is to minimize the number of blocks transmitted while satisfying heterogeneous reusability score constraints, it can be transformed into a series of feasibility problems: \textit{Does there exist $Z$-block subsets that make the task with worst assembled model achieves its QoS score?}
By increasing $Z$ sequentially, i.e., $Z=N, N+1, \cdots$, we can find the optimal $Z^*$ to make the optimized worst assembled model reach the score requirement.
By normalizing the blockwise task-specific reusability score, the max-min result should be no less than $1$.
The reduction operation is in polynomial time and the max-min optimization problem can be written as 
\begin{equation*}
    \begin{aligned}
        \max \limits_{\{\alpha_{i,n}^{(k)}\}} &  \quad \min_{k} \sum_{n=1}^N \sum_{i=1}^{K} \alpha_{i,n}^{(k)} I_{i,n}^{(k)},\\
        \mathrm{s.t.~} & \quad \sum_{i=1}^{K} \sum_{n=1}^N  \alpha_{i,n} \leq Z, \\
        & \quad \sum_{i=1}^K \alpha_{i,n}^{(k)} = 1, \forall n, k, \\
        & \quad \alpha_{i,n}^{(k)} \leq \alpha_{i,n}, \forall k, i, n, \\
        & \quad \alpha_{i, n}^{(k)} \in \{0, 1\}, \forall i, n, k, \\
        & \quad \alpha_{i, n} \in \{0, 1\}, \forall i, n.
    \end{aligned}
\end{equation*}
The block number constraint is equivalent to the volume constraint in MNK.
By setting one model block for a complete model, i.e., $N=1$, the problem is exactly the same as the classical MNK problem.
When $N$ is enlarged, the architecture constraint is to enforce picking exactly one item from each position group, the inherent max-min problem is not changed. 
Therefore, we conclude that the MNK, which according to \cite{Max-Min} is NP-complete, is a special case of the transformed minimum reusability score maximization problem.
Hence, problem (P2.1) is NP-complete.
This completes the proof.

\subsection{Proof of Lemma 2}
\label{appendix: B}

We apply the KKT conditions to show Lemma 2.
The Lagrange function of problem (P3) can be written as
\begin{equation*}
    \begin{aligned}
         \mathcal{L} = & \sum_{i=1}^K \sum_{n=1}^N T_{i,n} +  \beta (\sum_{i=1}^K\sum_{n=1}^N T_{i,n}p_{i,n} - E) \\ & + \sum_{i=1}^K\sum_{n=1}^N\sum_{k=1}^K \lambda_{i,n}^{(k)} (\alpha_{i,n}^{(k)} Q - T_{i,n}B\log_2(1 + p_{i,n}H_k/N_0)),
    \end{aligned}
\end{equation*}
where $\{\lambda_{i,n}^{(k)} \geq 0\}$, $\beta \geq 0$ are constraint-associated multipliers.
Then, KKT conditions are listed below:
\begin{equation*}
\left\{
    \begin{aligned}
        & \frac{\partial\mathcal{L}}{\partial T_{i,n}} = 1 - \sum_{k=1}^{K} \lambda_{i,n}^{(k)} B\log_2(1+\frac{p_{i,n}H_k}{N_0}) + \beta p_{i,n} = 0, \\
        & \frac{\partial \mathcal{L}}{\partial p_{i,n}} = - \frac{B}{\ln 2} T_{i,n} \sum_{k=1}^K \lambda_{i,n}^{(k)} \frac{p_{i,n}}{N_0 + p_{i,n}H_k} + \beta T_{i,n} = 0, \\
        & \lambda_{i,n}^{(k)} ({\alpha_{i,n}^{(k)}}^*Q - T_{i,n}B\log_2(1 + \frac{p_{i,n}H_k}{N_0})) = 0, \\
        & \beta (\sum_{i=1}^K \sum_{n=1}^N T_{i,n}p_{i,n} - E) = 0.
    \end{aligned}
\right.
\end{equation*}

Consider one model block $s_{i,n}$.
Given the obtained model block parameter scheme $\{ {\alpha_{i,n}^{(k)}}^*\}$, if $\forall k \in \mathcal{K}, {\alpha_{i,n}^{(k)}}^* = 0 $, then $T_{i,n}=0$.
Otherwise, if $ \exists k \in \mathcal{K}, {\alpha_{i,n}^{(k)}}^* = 1$, then $T_{i,n} \neq 0$.
In the latter case, for ${\alpha_{i,n}^{(k)}}^*$, $\lambda_{i,n}^{(k)} = 0$ should hold for most $k$, and there only exists one $k \in \mathcal{K}$ that ${\alpha_{i,n}^{(k)}}^* = 1$ with $\lambda_{i,n}^{(k)} \geq 0$ because of the existence of various channel condition in completeness slackness condition to guarantee the validity of the second condition.
Since ${\alpha_{i,n}^{(k)}}^*Q - T_{i,n}B\log_2(1 + \frac{p_{i,n}H_k}{N_0}) = 0$ holds for $\lambda_{i,n}^{(k)} \geq 0$, the device with $\lambda_{i,n}^{(k)} \geq 0$ should be the one with the worst channel state in all served devices, i.e., $\{{\alpha_{i,n}^{(k)}}^* = 1\}$.
Therefore, the original power control problem (P3) can be shifted toward blockwise power allocation stated in problem (P3.1).

Now, consider the blockwise power control.
Since $\beta \neq 0$, the power control should run out all the energy budget to minimize the downloading latency, i.e., $\sum_{i=1}^{K}\sum_{n=1}^N T_{i,n} p_{i,n} = E$.
And the blockwise power $p_{i,n}$ is determined by the associated non-negative multiplier $\lambda_{i,n}^{(k)} \geq 0$.
Hence, by substituting the relationship between blockwise power and latency, the multiplier $\lambda_{i,n}^{(k)}$ can be written as a function of $\beta$.
Thus, the optimal power control policy is deduced as shown in \eqref{eq: power control}.
And the multiplier $\beta$ can be determined by the usage of energy consumption, as the result given in \eqref{eq: multiplier}.
This completes the proof.

\bibliographystyle{IEEEtran}
\bibliography{Ref}

\begin{thebibliography}{10}
\providecommand{\url}[1]{#1}
\csname url@samestyle\endcsname
\providecommand{\newblock}{\relax}
\providecommand{\bibinfo}[2]{#2}
\providecommand{\BIBentrySTDinterwordspacing}{\spaceskip=0pt\relax}
\providecommand{\BIBentryALTinterwordstretchfactor}{4}
\providecommand{\BIBentryALTinterwordspacing}{\spaceskip=\fontdimen2\font plus
\BIBentryALTinterwordstretchfactor\fontdimen3\font minus
  \fontdimen4\font\relax}
\providecommand{\BIBforeignlanguage}[2]{{%
\expandafter\ifx\csname l@#1\endcsname\relax
\typeout{** WARNING: IEEEtran.bst: No hyphenation pattern has been}%
\typeout{** loaded for the language `#1'. Using the pattern for}%
\typeout{** the default language instead.}%
\else
\language=\csname l@#1\endcsname
\fi
#2}}
\providecommand{\BIBdecl}{\relax}
\BIBdecl

\bibitem{Whitepaper_Huawei}
W.~Tong and P.~Zhu, Eds., \emph{6G: The Next Horizon: From Connected People and
  Things to Connected Intelligence}.\hskip 1em plus 0.5em minus 0.4em\relax
  Cambridge University Press, 2021.

\bibitem{6G_Letaief}
K.~B. Letaief, Y.~Shi, J.~Lu, and J.~Lu, ``Edge artificial intelligence for
  {6G}: Vision, enabling technologies, and applications,'' \emph{IEEE J. Sel.
  Areas Commun.}, vol.~40, no.~1, pp. 5--36, 2021.

\bibitem{Dis_learning_Chen}
M.~Chen, D.~Gündüz, K.~Huang, W.~Saad, M.~Bennis, A.~V. Feljan, and H.~V.
  Poor, ``Distributed learning in wireless networks: Recent progress and future
  challenges,'' \emph{IEEE J. Sel. Areas Commun.}, vol.~39, no.~12, pp.
  3579--3605, 2021.

\bibitem{Split_inf_Shao}
J.~Shao and J.~Zhang, ``Communication-computation trade-off in
  resource-constrained edge inference,'' \emph{IEEE Commun. Mag.}, vol.~58,
  no.~12, pp. 20--26, 2020.

\bibitem{Model_downloading_KB}
K.~Huang, H.~Wu, Z.~Liu, and X.~Qi, ``In-situ model downloading to realize
  versatile edge {AI} in {6G} mobile networks,'' \emph{to appear in IEEE
  Wireless Commun.}, 2023.

\bibitem{3GPP}
\BIBentryALTinterwordspacing
J.~Shen, ``{5G System (5GS); Study on traffic characteristics and performance
  requirements for AI/ML model transfer},'' {3rd Generation Partnership Project
  (3GPP)}, Technical Specification (TS) 22.874, June~21 2021. [Online].
  Available:
  \url{https://portal.3gpp.org/desktopmodules/Specifications/SpecificationDetails.aspx?specificationId=3721}
\BIBentrySTDinterwordspacing

\bibitem{JSCC_Deniz}
T.-Y. Tung, D.~B. Kurka, M.~Jankowski, and D.~Gündüz, ``{DeepJSCC-Q}:
  Constellation constrained deep joint source-channel coding,'' \emph{IEEE J.
  Sel. Areas Inf. Theory}, 2022.

\bibitem{Device_edge_Zhou}
W.~Shi, Y.~Hou, S.~Zhou, Z.~Niu, Y.~Zhang, and L.~Geng, ``Improving device-edge
  cooperative inference of deep learning via 2-step pruning,'' in \emph{Proc.
  IEEE Conf. Comput. Commun. Workshops (INFOCOM WKSHPS)}, Paris, France,
  April~29 -- May~2 2019, pp. 1--6.

\bibitem{Edge_AI_Li}
E.~Li, L.~Zeng, Z.~Zhou, and X.~Chen, ``Edge {AI}: On-demand accelerating deep
  neural network inference via edge computing,'' \emph{IEEE Trans. Wireless
  Commun.}, vol.~19, no.~1, pp. 447--457, 2019.

\bibitem{Cooper_inference_Shao}
J.~Shao, Y.~Mao, and J.~Zhang, ``Task-oriented communication for multidevice
  cooperative edge inference,'' \emph{IEEE Trans. Wireless Commun.}, vol.~22,
  no.~1, pp. 73--87, 2023.

\bibitem{Bottleneck++_Shao}
J.~Shao and J.~Zhang, ``Bottlenet++: An end-to-end approach for feature
  compression in device-edge co-inference systems,'' in \emph{Proc. IEEE Int.
  Conf. Commun. Workshops (ICC Workshops)}, Virtual Event, June~7--11 2020, pp.
  1--6.

\bibitem{Wireless_AI_Park}
J.~Park, S.~Samarakoon, M.~Bennis, and M.~Debbah, ``Wireless network
  intelligence at the edge,'' \emph{Proc. IEEE}, vol. 107, no.~11, pp.
  2204--2239, 2019.

\bibitem{PFTX_Lan}
Q.~Lan, Q.~Zeng, P.~Popovski, D.~Gündüz, and K.~Huang, ``Progressive feature
  transmission for split classification at the wireless edge,'' \emph{IEEE
  Trans. Wireless Commun.}, 2022.

\bibitem{Batch_Liu}
Z.~Liu, Q.~Lan, and K.~Huang, ``Resource allocation for multiuser edge
  inference with batching and early exiting,'' \emph{IEEE J. Sel. Areas
  Commun.}, vol.~41, no.~4, pp. 1186--1200, 2023.

\bibitem{Compression_NN_Li}
J.~Guo, J.~Wang, C.-K. Wen, S.~Jin, and G.~Y. Li, ``Compression and
  acceleration of neural networks for communications,'' \emph{IEEE Commun.
  Mag.}, vol.~27, no.~4, pp. 110--117, 2020.

\bibitem{Model_placement_Yan}
J.~Yan, S.~Bi, and Y.-J.~A. Zhang, ``Optimal model placement and online model
  splitting for device-edge co-inference,'' \emph{IEEE Trans. Wireless
  Commun.}, vol.~21, no.~10, pp. 8354--8367, 2022.

\bibitem{Foundation_Model_2021}
\BIBentryALTinterwordspacing
R.~Bommasani, D.~A. Hudson, E.~Adeli \emph{et~al.}, ``On the opportunities and
  risks of foundation models,'' Standford University, Tech. Rep., 2022.
  [Online]. Available: \url{https://crfm.stanford.edu/assets/report.pdf}
\BIBentrySTDinterwordspacing

\bibitem{Pathway}
\BIBentryALTinterwordspacing
D.~Jeff. (2021) Introducing pathways: A next-generation {AI} architecture.
  [Online]. Available:
  \url{https://blog.google/technology/ai/introducing-pathways-next-generation-ai-architecture/}
\BIBentrySTDinterwordspacing

\bibitem{TL_survey}
F.~Zhuang, Z.~Qi, K.~Duan, D.~Xi, Y.~Zhu, H.~Zhu, H.~Xiong, and Q.~He, ``A
  comprehensive survey on transfer learning,'' \emph{Proc. IEEE}, vol. 109, pp.
  43--76, 2020.

\bibitem{Reassembly_2022}
X.~Yang, Z.~Daquan, S.~Liu, J.~Ye, and X.~Wang, ``Deep model reassembly,'' in
  \emph{Proc. Adv. Neural Inform. Process. Syst. (NeurIPS)}, New Orleans, LA,
  USA, Nov~28--Dec~9 2022.

\bibitem{SNNet_2023}
Z.~Pan, J.~Cai, and B.~Zhuang, ``Stitchable neural networks,'' in \emph{Proc.
  IEEE Conf. Comput. Vis. Pattern Recognit. (CVPR)}, Vancouver B.C. Canada,
  Jun~18--22 2023.

\bibitem{muNet_2022}
A.~Gesmundo and J.~Dean, ``An evolutionary approach to dynamic introduction of
  tasks in large-scale multitask learning systems,'' \emph{arXiv preprint
  arXiv:2205.12755}, 2022.

\bibitem{Shapley_Game_Theory_2002}
E.~Winter, ``The {Shapley} value,'' \emph{Handbook of game theory with economic
  applications}, vol.~3, pp. 2025--2054, 2002.

\bibitem{ResNet_2016}
K.~He, X.~Zhang, S.~Ren, and J.~Sun, ``Deep residual learning for image
  recognition,'' in \emph{Proc. IEEE Conf. Comput. Vis. Pattern Recognit.
  (CVPR)}, Las Vegas, NV, USA, June~26--July~1 2016.

\bibitem{ViT_2021}
A.~Dosovitskiy, L.~Beyer, A.~Kolesnikov, D.~Weissenborn, X.~Zhai,
  T.~Unterthiner, M.~Dehghani, M.~Minderer, G.~Heigold, S.~Gelly, J.~Uszkoreit,
  and N.~Houlsby, ``An image is worth 16x16 words: Transformers for image
  recognition at scale,'' in \emph{Proc. Int. Conf. Learn. Represent. (ICLR)},
  Vienna, Austria, May~3--7 2021.

\bibitem{XAI_2019}
D.~Gunning, M.~Stefik, J.~Choi, T.~Miller, S.~Stumpf, and G.-Z. Yang, ``{XAI}
  — explainable artificial intelligence,'' \emph{Science robotics}, vol.~4,
  no.~37, p. eaay7120, 2019.

\bibitem{Explain_Shapley_2019}
M.~Ancona, C.~Oztireli, and M.~Gross, ``Explaining deep neural networks with a
  polynomial time algorithm for {Shapley} value approximation,'' in \emph{Proc.
  Int. Conf. Mach. Learn. (ICML)}, Long Beach, CA, USA, Jun~9--15 2019.

\bibitem{Neuron_Shapley_2020}
A.~Ghorbani and J.~Y. Zou, ``Neuron {Shapley}: Discovering the responsible
  neurons,'' in \emph{Proc. Adv. Neural Inform. Process. Syst. (NeurIPS)},
  Vancouver, BC, Canada, Dec~6--12 2020.

\bibitem{Shapley_Value_2020}
M.~Sundararajan and A.~Najmi, ``The many {Shapley} values for model
  explanation,'' in \emph{Proc. Int. Conf. Mach. Learn. (ICML)}, Virtual Event,
  Jul~13--18 2020.

\bibitem{power_control}
M.~Chiang, P.~Hande, T.~Lan, C.~W. Tan \emph{et~al.}, ``Power control in
  wireless cellular networks,'' \emph{Foundations and Trends{\textregistered}
  in Networking}, vol.~2, no.~4, pp. 381--533, 2008.

\bibitem{Relaxation}
C.~Y. Wong, R.~Cheng, K.~Lataief, and R.~Murch, ``Multiuser {OFDM} with
  adaptive subcarrier, bit, and power allocation,'' \emph{IEEE J. Sel. Areas
  Commun.}, vol.~17, no.~10, pp. 1747--1758, 1999.

\bibitem{cvx}
M.~Grant and S.~Boyd, ``{CVX}: Matlab software for disciplined convex
  programming, version 2.1,'' \url{http://cvxr.com/cvx}, Mar. 2014.

\bibitem{ImageNet_Russakovsky}
O.~Russakovsky, J.~Deng, H.~Su, J.~Krause, S.~Satheesh, S.~Ma, Z.~Huang,
  A.~Karpathy, A.~Khosla, M.~Bernstein \emph{et~al.}, ``{ImageNet} large scale
  visual recognition challenge,'' \emph{Int. J. Comput. Vis.}, vol. 115, no.~3,
  pp. 211--252, 2015.

\bibitem{CIFAR_Krizhevsky}
A.~Krizhevsky, G.~Hinton \emph{et~al.}, ``Learning multiple layers of features
  from tiny images,'' 2009.

\bibitem{SVHN}
Y.~Netzer, T.~Wang, A.~Coates, A.~Bissacco, B.~Wu, and A.~Y. Ng, ``Reading
  digits in natural images with unsupervised feature learning,'' in \emph{Proc.
  of Int. Conf. Neural Inform. Process. Syst. Workshops (NIPS Workshops)},
  2011.

\bibitem{Max-Min}
G.~Yu, ``On the max-min 0-1 knapsack problem with robust optimization
  applications,'' \emph{Operations Research}, vol.~44, no.~2, pp. 407--415,
  1996.

\end{thebibliography}

\end{document}